\documentclass[journal, a4paper, twoside]{IEEEtran}

\usepackage{amsmath,amssymb,graphicx,epsfig, cite,algorithm,algorithmic,epstopdf, url}
\usepackage[utf8]{inputenc}
\usepackage[mathscr]{euscript}
\usepackage{hyperref}
\usepackage{color}
\usepackage{amsthm}
\usepackage{bbm}
\usepackage{booktabs}

\usepackage{algorithm}
\usepackage{algorithmic}

\usepackage{siunitx}
\def\minwrt[#1]{\underset{#1}{\text{minimize }}}
\def\argminwrt[#1]{\underset{#1}{\text{arg min }}}
\def\maxwrt[#1]{\underset{#1}{\text{maximize }}}
\def\maxemphwrt[#1]{\underset{#1}{\text{\emph{maximize} }}}

\def\supwrt[#1]{\underset{#1}{\text{sup }}}

\newtheorem{theorem}{Theorem}
\newtheorem{remark}{Remark}
\newtheorem{proposition}{Proposition}

\newcommand{\costf}{c}

\newcommand{\tr}{{\rm tr}}

\newcommand{\indfunc}[1]{\mathbbm{1}_{\left\{#1\right\}}}
\newcommand{\norm}[1]{\left\lVert#1\right\rVert}
\newcommand{\abs}[1]{\left|#1\right|}
\newcommand{\trace}[1]{\text{tr}\left(#1\right)}
\def\ba{{\bf a}}

\def\bI{{\bf I}}

\def\bx{{\bf x}}

\def\bz{{\bf z}}

\def\bR{{\bf R}}
\def\bQ{{\bf Q}}

\def\bT{{\bf T}}
\def\bU{{\bf U}}

\def\bX{{\bf X}}
\def\bY{{\bf Y}}

\def\ccM{{\mathcal{M}}}
\def\ccI{{\mathcal{I}}}

\def\RE{{\mathbb{E}}}

\def\RM{{\mathbb{M}}}
\def\RN{{\mathbb{N}}}
\def\RR{{\mathbb{R}}}
\def\RT{{\mathbb{T}}}
\def\RZ{{\mathbb{Z}}}

\newcommand{\bLambda}{\boldsymbol{\Lambda}}

\newcommand{\Omtspec}{S}
\newcommand{\Omtcov}{T}
\newcommand{\costfunc}{c}

        \makeatletter
        \def\fps@eqnfloat{!t}
        \def\ftype@eqnfloat{4}
        
        \newenvironment{eqnfloat*}
               {\@dblfloat{eqnfloat}}
               {\end@dblfloat}
        \makeatother
\renewenvironment{thebibliography}[1]{%
 \begin{oldthebibliography}{#1}%
 \setlength{\parskip}{0ex}%
 \setlength{\itemsep}{0ex}%
}%
{%
\end{oldthebibliography}%
}%

\title{Interpolation and Extrapolation of Toeplitz Matrices via Optimal Mass Transport}
\author{Filip Elvander, \textit{Student Member, IEEE}, Andreas Jakobsson, \textit{Senior Member, IEEE},\\and Johan Karlsson, \textit{Senior Member, IEEE}\thanks{This work was supported in part by the Swedish Research Council.}\thanks{Parts of the material herein have been published at the IEEE International Conference on Acoustics, Speech and Signal Processing, Calgary, Canada, April 2018.}
\thanks{F.~Elvander and A.~Jakobsson are with the Center for Mathematical Sciences,
 Lund University, SE-22100 Lund, Sweden (email:~filip.elvander@matstat.lu.se, aj@maths.lth.se). J.~Karlsson is with the Department of Mathematics, KTH Royal Institute of Technology, Stockholm, Sweden (email:~johan.karlsson@math.kth.se).}}

\begin{document}
\maketitle
%%%%%% Abstract %%%%%%%%%%%
%%%%%%%%%%%%%%%%%%%%%%

%%%%%%%%%%%%%%%%%%%%%%
\begin{abstract}
In this work, we propose a novel method for quantifying distances between Toeplitz structured covariance matrices. By exploiting the spectral representation of Toeplitz matrices, the proposed distance measure is defined based on an optimal mass transport problem in the spectral domain. This may then be interpreted in the covariance domain, suggesting a natural way of interpolating and extrapolating Toeplitz matrices, such that the positive semi-definiteness and the Toeplitz structure of these matrices are preserved. The proposed distance measure is also shown to be contractive with respect to both additive and multiplicative noise, and thereby allows for a quantification of the decreased distance between signals when these are corrupted by noise. Finally, we illustrate how this approach can be used for several applications in signal processing. In particular, we consider interpolation and extrapolation of Toeplitz matrices, as well as clustering problems and tracking of slowly varying stochastic processes.
\end{abstract}
%%%%%%%%%%%%%%%%%%%%%%%%%%%%
%\begin
{\bf Keywords}
Covariance interpolation, Optimal mass transport, Toeplitz matrices, Spectral estimation

\section{Introduction}%\vspace{-2mm}
%\vspace{-0.28cm}
%
Statistical modeling is a key methodology for estimation and identification and is used throughout the signal processing field. An intrinsic component of such models is covariance estimates, which are extensively used in application areas such as spectral estimation, radar, and sonar \cite{Burg67, StoicaM05}, wireless channel estimation, medical imaging, and identification of systems and network structures \cite{Dempster72_28, ChandrasekaranPW12_40}. Although being a classical subject (see, e.g., \cite{Kshirsagar72}), covariance estimation has recently received considerable attention. Such contributions include works on finding robust covariance estimates with respect to outliers, as well as methods suitable for handling different distribution assumptions, including families of non-Gaussian distributions  \cite{ZoubirKCM12_29,OllilaTKP12_60,OllilaT14_62, Wiesel12_60, SunBP16_64}.
Another important class of problems is covariance estimation with an inherent geometry that gives rise to a structured covariance matrix. 
Such structures could arise from stationarity assumptions of the underlying object  
\cite{BurgLW82_70, ByrnesGL00_48, LiSL99_47, Georgiou03_book, KarlssonE12_mtns} or be due to assumptions in, e.g., the underlying network structures in graphical models \cite{AvventiLB13_58,ZorziS16_61}.
In this work, we focus on Toeplitz structures which naturally arise when modeling stationary signals and processes.

Although many methods rely on stationarity for modeling signals, such assumptions are typically not valid over longer time horizons. Therefore, tools for interpolation and morphing of covariance matrices are important for modeling and fusing information. A straightforward and often used approach is the Euclidean metric; however, this metric does not take into account the underlying geometry and typically results in fade-in-fade-out effects (as is also illustrated herein). Several other such tools for interpolating covariances have recently been proposed in the literature, for example methods based on g-convexity  \cite{Wiesel12_60}, optimal mass transport \cite{YamamotoCNGT16_arxiv}, and information geometry \cite{Barbaresco13}. An alternative approach for such interpolation is to relax, or "lift'',  the covariances and instead consider interpolation between the lifted objects. For example, in \cite{KnottS84_43} (see also \cite{NingJG13_20}), 
interpolation between covariance matrices is induced from the optimal mass transport geodesics between the Gaussian 
density functions with the corresponding covariances. However, neither of these interpolation approaches take into account that the covariance matrix represents an (almost) stationary times series and do not preserve the Toeplitz structure of the interpolating covariance sequence.

The topic of optimal mass transport (see, e.g., \cite{Villani08,KolouriPTSR17_34}) was originally introduced in order to address the problem of, in a cost efficient way, supplying construction sites with building material and has been used in many contexts, such as, e.g., economics and resources allocation. Lately, it has also gained interest in application fields such as image processing \cite{HakerZTA04_60, DeGoesCSAD11_30} signal processing \cite{EngquistF14_12, GeorgiouKT09_57, JiangLG12_60}, computer vision and machine learning \cite{RubnerTG00_40, LingO07_29, RabinPDB11_conf, ArjovskyCB17_arXiv,AdlerROK17_arxiv}. In this work, we will utilize optimal mass transport to model changes in the covariance structure of stochastic processes, or signals. To this end, we propose a new lifting approach, where the lifting is made from the covariance domain to the frequency domain, using the fact that any positive semi-definite Toeplitz matrix has a spectral representation. 
We combine this approach with the frequency domain metric based on optimal mass transport, proposed in \cite{GeorgiouKT09_57}, in order to define pairwise distances between Toeplitz matrices.
This is done by considering the minimum distance, in the optimal mass transport sense, between the sets of power spectra consistent with each of the Toeplitz matrices. The proposed distance measure is shown to be contractive with respect to additive and multiplicative noise, i.e., it reflects the increased difficulty of discriminating between two stochastic processes if these are corrupted by two realizations of a noise process. 
Also, we show that the proposed distance measure gives rise to a natural way of interpolating and extrapolating Toeplitz matrices. The interpolation method preserves the Toeplitz structure, the positive semi-definiteness, as well as the diagonal of the interpolating/extrapolating matrices.

The proposed optimal mass transport problem is in its original form an infinite-dimensional problem. As an alternative to finding solutions using approximations based on discretizations of the underlying space, we show that certain formulations of the problem allows for approximations by a semi-definite program using a sum-of-squares representation. Also, we illustrate how the method can be used for interpolation, extrapolation, tracking, and clustering.  

This paper is organized as follows. In Section~\ref{sec:background}, we provide a brief background on the moment problem, i.e., determining the power spectrum from a partial covariance sequence or finite covariance matrix, as well as introduce the problem of optimal mass transport. Section~\ref{seq:distance_notion} introduces the proposed distance notion for positive semi-definite Toeplitz covariance matrices. Here, the dual problem is derived, and properties of the proposed distance notion are described. In Section~\ref{sec:applications}, we describe applications of the proposed distance notion, such as induced interpolation, extrapolation, tracking, and clustering. Section~\ref{sec:sos} formulates a sum-of-squares relaxation of the dual problem. Section~\ref{sec:num} provides numerical illustrations of the proposed distance notion, as well as the described applications. Finally, Section~\ref{sec:conclusions} concludes upon the work.
%

%
%%%%% NOTATION %%%%%%%
\subsection*{Notation}
%\vspace{-0.3cm}
%
Let $\RM^{n}$ denote the set of Hermitian $n\times n$ matrices and let  $(\cdot)^T$ denote the transpose, $(\cdot)^H$ the Hermitian transpose, and $\overline{(\cdot)}$ the complex conjugate.
Let $\RT=(-\pi,\pi]$ and let $C_{\rm perio}(\RT)$ denote the set of continuous and $2\pi$-periodic functions on $\RT$. The set of linear bounded functionals on $C_{\rm perio}(\RT)$, or equivalently, the dual space of $C_{\rm perio}(\RT)$, is the set of generalized integrable functions on the set $\RT$, here denoted by $\ccM(\RT)$. Thus, $\ccM(\RT)$ includes, e.g., functions containing singular parts such as Dirac delta functions \cite{Luenberger69}.%
\footnote{Strictly speaking $\ccM(\RT)$ is the set of signed bounded measures on $\RT$ \cite{Friedman70}.}
Further, we
let $\ccM_+(\RT)$ denote the subset of such functions that are non-negative.
We use $\langle \Phi, f \rangle$ to denote the application of the functional $\Phi$ on $f$, e.g., 
\begin{align*}
	\langle \Phi, f \rangle = \int_{\RT} f(\theta) \Phi(\theta)  d\theta
\end{align*}
if $f\in C_{\rm perio}(\RT)$ and $\Phi\in \ccM(\RT)$. For Hilbert spaces, $\langle \cdot, \cdot \rangle$ is the standard inner product, e.g., when $\bX$ and $\bY$ are vectors or matrices, then  
$\langle \bX,\bY\rangle=\tr(\bX\bY^H)$, where $\tr(\cdot)$ denotes the trace. We denote matrices by boldface upper-case letters, such as $\bX$, whereas vectors are denoted by boldface lower-case letters, such as $\bx$.
Furthermore, $\norm{\bX}_F = \sqrt{\langle \bX,\bX\rangle}$ denotes the Frobenius norm induced by the matrix inner product. Lastly, for $f \in C_{\rm perio}(\RT)$, we let
\begin{align*}
	\norm{f}_1 = \supwrt[\abs{\Phi(\theta)}\leq 1] \langle \Phi, f \rangle =  \int_\RT \abs{f(\theta)}d\theta
\end{align*}
denote the $L_1$-norm.
%

%
%
%
%
%%%%%%%% BACKGROUND %%%%%%%%
\section{Background}\label{sec:background}
\subsection{Stochastic processes and spectral representations} \label{ssec:background_stoch}
We will in this work consider complex-valued discrete time stochastic processes, or signals, $y(t)$ for $t \in \RZ$. These will be assumed to be zero mean and wide sense stationary (WSS), i.e., $\RE (y(t))=0$  for all $t\in \RZ$, and the covariance 
\begin{equation} \label{eq:r_k}
r_k \triangleq \RE (y(t)\overline{y(t-k)})
\end{equation}
being independent of $t$. Here, $\RE(\cdot)$ denotes the expectation operator. The frequency content of the process $y(t)$ may then be represented by the power spectrum, $\Phi$, i.e., the non-negative function on $\RT$ whose Fourier coefficients coincide with the covariances:
\begin{equation}\label{eq:moments}
r_k=\frac{1}{2\pi}\int_{-\pi}^\pi \Phi(\theta)e^{-ik\theta} d\theta
\end{equation}
for $k\in \RZ$  (see, e.g.,  \cite[Chapter 2]{JohnsonD93}). Typically in spectral estimation, one considers the inverse problem of recovering the power spectrum $\Phi$ from a given set of covariances $r_k$, for $k \in \RZ,$ with $\abs{k}\leq n-1$. The condition for any such reconstruction to be valid is that $\Phi$ should be consistent with the covariance sequence $\left\{ r_k \right\}_{\abs{k}\leq n-1}$, i.e., \eqref{eq:moments} should hold for $\abs{k}\leq n-1$. The corresponding $n\times n$ covariance matrix is

\begin{align}\label{eq:Toeplitz}
\bR=\begin{pmatrix}
r_0&r_{-1}& r_{-2}&\cdots &r_{-n+1}\\
r_1&r_0& r_{-1}&\cdots &r_{-n+2}\\
r_2&r_1& r_{0}&\cdots &r_{-n+3}\\
\vdots &\vdots& \vdots&\ddots&\vdots\\
r_{n-1}& r_{n-2}&r_{n-3}&\cdots& r_{0}
\end{pmatrix}
\end{align}
which is a Hermitian Toeplitz matrix, since $y(t)$ is WSS. Thus, expressed in the form of matrices, a spectrum is consistent with an observed partial covariance sequence, or, equivalently, a finite covariance matrix, if $\Gamma(\Phi) = \bR$,
where \mbox{$\Gamma:\ccM(\RT) \to \RM^{n}$} is the linear operator 
\begin{align}
	\Gamma(\Phi)\triangleq\frac{1}{2\pi}\int_{\RT} \ba(\theta)\Phi(\theta) \ba(\theta)^Hd\theta
\end{align}
and
\begin{align}
	\ba(\theta) \triangleq \left[\begin{array}{cccc} 1 & e^{i\theta} & \cdots &  e^{i(n-1)\theta} \end{array} \right]^T
\end{align}
is the Fourier vector. Note that $\Gamma(\Phi)$ is a Toeplitz matrix, since $\ba(\theta)\ba(\theta)^H$ is Toeplitz for any $\theta$.
It may be noted that for any positive semi-definite Toeplitz matrix, $\bR$, there always exists at least one consistent power spectrum; in fact, if $\bR$ is positive definite, there is an infinite family of consistent power spectra
 \cite{GrenanderS58}. It may be noted that for singular Toeplitz covariance matrices, the spectral representation is unique. This fact has recently been successfully utilized in atomic norm minimization problems for grid-less compressed sensing of sinusoidal signals (see, e.g., \cite{BhaskarTR13_61,ElvanderSJ18_145}). In this work, we are mainly interested in the non-singular case, where several power spectra are consistent with given covariance matrices. In Section \ref{seq:distance_notion}, we will utilize such spectral representations in order to define a notion of distance between pairs of Toeplitz matrices. This distance will be defined in terms of the minimum optimal mass transport cost between the sets of power spectra consistent with the matrices.
 %
%
%%%%%%%%% OMT %%%%%%%%%%%%
\subsection{Optimal mass transport}
The Monge-Kantorovich problem of optimal mass transport is the problem of finding an optimal transport plan between two given mass distributions \cite{Villani08,KolouriPTSR17_34}. The cost of moving a unit mass is defined on the underlying space, and the optimal transport plan is defined as the plan minimizing the total cost.
The resulting minimal cost, associated with the optimal transport plan, can then be used as a measure of similarity, or distance, between the two mass distributions. 
The idea of utilizing the optimal mass transport cost as a distance measure has been used, e.g., for defining metrics on the space of power spectra \cite{GeorgiouKT09_57}, whereas the optimal transport plan has been used for tracking stochastic processes with smoothly varying spectral content and for spectral morphing for speech signals \cite{JiangLG12_60}. Recently, the interpretation of the optimal transport plan as providing an optimal association between elements in two mass distributions has been used as a means of clustering in fundamental frequency estimation algorithms \cite{ElvanderAKJ17_icassp}.
One of the advantages of using the optimal mass transport as a distance compared to traditional metrics is that it naturally incorporates the geometry of the underlying space. In particular, the optimal mass transport cost between two objects depends on the distance between the two objects in the underlying space, whereas standard metrics only depend on the overlapping regions. Furthermore, interpolation using optimal mass transport results in smooth transitions in the underlying space (see Sections \ref{ex:DOA} and \ref{ex:tracking}). This makes optimal mass transport suitable for applications where there is a smooth transition in the underlying space, e.g., tracking problems in radar and sonar.

As in \cite{GeorgiouKT09_57}, we consider the following distance between two spectra $\Phi_0$ and $\Phi_1$:
\begin{subequations}\label{eq:omt}
\begin{align} 
\Omtspec(\Phi_0,\Phi_1)\;\triangleq\min_{M\in \ccM_+(\RT^2)}\qquad & \int_{\RT^2} \costfunc(\theta,\varphi)M(\theta,\varphi)d\theta d\varphi\label{eq:omt_a}\\
\mbox{ subject to }\quad &\Phi_0(\theta)=\int_\RT M(\theta,\varphi)d\varphi\label{eq:omt_b}\\
&\Phi_1(\varphi)=\int_\RT M(\theta,\varphi)d\theta\label{eq:omt_c}
\end{align}
\end{subequations}
where $\RT^2 = \RT\times\RT$ denotes the 2-D frequency space. Here, the cost function, $\costfunc(\theta,\varphi)$, details the cost of moving one unit of mass between the frequencies $\theta$ and $\varphi$.
 The transport plan, $M(\theta,\varphi)$, specifies the amount of mass moved from frequency $\theta$ to frequency $\varphi$. The objective in \eqref{eq:omt_a} is the total cost associated with the transport plan $M$ and the constraints \eqref{eq:omt_b} and \eqref{eq:omt_c} ensure that $M$ is a valid transport plan from  $\Phi_0$ to $\Phi_1$, i.e., the integration marginals of $M$ coincide with the spectra $\Phi_0$ and $\Phi_1$. It may be noted that, due to these marginal constraints, the distance measure $\Omtspec$ is only defined for spectra of the same mass, or total power. However, $\Omtspec$ may be generalized in order to allow for mass differences by including a cost for adding and subtracting mass by postulating that the spectra $\Phi_0$ and $\Phi_1$  are perturbations of functions $\Psi_0$ and $\Psi_1$ that have equal mass. As in \cite{GeorgiouKT09_57}, this may be formulated as
\begin{align}\label{eq:Skappa}
\Omtspec_\kappa(\Phi_0,\Phi_1)=\min_{\Psi_j\in \ccM_+(\RT)}\; & \Omtspec(\Psi_0,\Psi_1)+\kappa \sum_{j=0}^1\norm{\Phi_j - \Psi_j}_1
\end{align}
where $\kappa>0$ is a used-defined parameter detailing the cost of adding or subtracting mass. 
One interpretation of this is that points that are close represent the same object and can thus be transported via the first term in \eqref{eq:Skappa}, whereas points that are far apart represent different objects and must be phased in/out using the second term in \eqref{eq:Skappa}. Then, $\kappa$ may be interpreted as a parameter determining when two points are close.
If the cost function in the optimal mass transport problem in \eqref{eq:omt} is chosen as $\costfunc(\theta,\phi)=d(\theta,\phi)^p$, $p\geq 1$ for any metric $d(\theta,\phi)$ on $\RT$, then $W(\Phi_0,\Phi_1)=\Omtspec(\Phi_0,\Phi_1)^{1/p}$ is the so-called Wasserstein metric on $\ccM_+(\RT)$. Similarly, for $S_\kappa$, the following theorem holds.
\begin{theorem}[\!\!\cite{GeorgiouKT09_57}]
Let $p\ge 1$, $\kappa>0$, and let the cost function be $\costfunc(\theta,\phi)=|\theta-\phi|^p$. Then,  
\begin{equation}
\label{eq:Wasserstein}
W_\kappa(\Phi_0,\Phi_1)=S_\kappa(\Phi_0,\Phi_1)^{1/p}
\end{equation}
is a metric on $\ccM_+(\RT)$.
\end{theorem}
Consider a situation where we need to discriminate between two
signals on the basis of their statistics or of their power spectra. In
such cases, additive or multiplicative noise typically impede our
ability to differentiate between the two. In particular,  we have that $\Phi \mapsto \Phi+\Phi_{\rm a}$ represents the operation of adding independent noise with spectrum $\Phi_{\rm a}$ and $\Phi \mapsto \Phi*\Phi_{\rm m}$ represents the operation of multiplying the signal with independent noise with spectrum $\Phi_{\rm m}$.
This was considered in \cite{GeorgiouKT09_57} and it was shown that the transportation distance respects this property in the sense that corrupting two signals with additive and (normalized) multiplicative noise decreases their transportation distance. Specifically, the following theorem holds.
\begin{theorem}[\!\!\cite{GeorgiouKT09_57}] \label{theorem:georgiou_contractive}
Let $p\ge 1$, $\kappa>0$, and let the cost function be $\costfunc(\theta,\phi)=|\theta-\phi|^p$, and let $W_\kappa(\Phi_0,\Phi_1)$ be defined by \eqref{eq:Wasserstein}.
Then, $W_\kappa(\Phi_0,\Phi_1)$ is contractive with respect to the additive and normalized multiplicative noise, i.e.,
\begin{itemize}
\setlength{\itemsep}{1pt}
\setlength{\parskip}{1pt}
\item $W_\kappa(\Phi_0+\Phi_{\rm a},\Phi_1+\Phi_{\rm a})\le W_\kappa(\Phi_0,\Phi_1)$
\item $W_\kappa(\Phi_0*\Phi_{\rm m},\Phi_1*\Phi_{\rm m})\le W_\kappa(\Phi_0,\Phi_1)$,
\end{itemize}
for any $\Phi_{\rm a}, \Phi_{\rm m}\in \ccM_+(\RT)$, with $\int_\RT \Phi_{\rm m}(\theta)d\theta=1$.
\end{theorem}
As we shall see, it is possible to construct notions of distance on the space of positive semi-definite Toeplitz matrices that have properties similar to those stated in Theorem~\ref{theorem:georgiou_contractive}.
 
%
%
%%%%%%%% NOTION OF DISTANCE FOR TOEPLITZ MATRICES %%%%%%%%
\section{A notion of distance for Toeplitz matrices} \label{seq:distance_notion}

As noted above, any positive semi-definite Toeplitz matrix $\bR$ has at least one spectral representation, i.e., there exists at least one spectrum $\Phi$ that is consistent with it. Thus, we define the distance, $\Omtcov$, between two positive semi-definite Toeplitz matrices, $\bR_0$ and $\bR_1$, as the minimum transportation cost, as measured by $\Omtspec$, between spectra consistent with the respective matrices, i.e.,
\begin{equation} \label{eq:omt_toeplitz}
\begin{aligned}
\Omtcov(\bR_0,\bR_1)\triangleq\min_{\Phi_0,\Phi_1\in \ccM_+(\RT)}\qquad & \Omtspec(\Phi_0,\Phi_1)\\
\mbox{ subject to }\qquad &\Gamma(\Phi_j)=\bR_j,\; j=0,1.
\end{aligned}
\end{equation}
Considering the definition of $\Omtspec$ in \eqref{eq:omt}, this can equivalently be formulated as the convex optimization problem
\begin{equation}\label{eq:omt_toeplitz2}
\begin{aligned}
 \Omtcov(\bR_0,\bR_1) = \min_{M\in \ccM_+(\RT^2)}\; &\int_{\RT^2} \costfunc(\theta,\varphi)M(\theta,\varphi)d\theta d\varphi\\
\mbox{ subject to }\; &\Gamma\left(\int_\RT M(\theta,\varphi)d\varphi\right) = \bR_0 \\
&\Gamma\left(\int_\RT M(\theta,\varphi)d\theta\right) = \bR_1.
\end{aligned}
\end{equation}
Note that the formulation in \eqref{eq:omt_toeplitz2} is only defined for covariance matrices with the same diagonal, i.e., the same $r_0$, as defined in \eqref{eq:r_k}, or, equivalently, covariance matrices whose consistent spectra have the same mass. 
However, in order to allow for mass differences, $\Omtcov(\bR_0,\bR_1)$ can be generalized in analog with \eqref{eq:Skappa} as
\begin{equation} \label{eq:omt_toeplitz_kappa}
\begin{aligned}
\Omtcov_\kappa(\bR_0,\bR_1)\triangleq\min_{\Phi_j\in \ccM_+(\RT)}\quad & \Omtspec_\kappa(\Phi_0,\Phi_1)\\
\mbox{ subject to }\quad &\Gamma(\Phi_j)=\bR_j \mbox{ for } j=0,1,
\end{aligned}
\end{equation}
or, equivalently,
\begingroup
\small
\begin{equation} \label{eq:omt_toeplitz_kappa2}
\begin{aligned}
\Omtcov_\kappa(\bR_0,\bR_1)=\min_{\substack{M \in \ccM_+(\RT^2)\\ \Psi_0, \Psi_1\in \ccM_+(\RT)}}\; & \int_{\RT^2}\costfunc(\theta,\varphi)M(\theta,\varphi)d\theta d\varphi\\
&+ \kappa\norm{\int_\RT M(\theta,\varphi)d\varphi - \Psi_0}_1 \\
&+ \kappa\norm{\int_\RT M(\theta,\varphi)d\theta - \Psi_1}_1 \\
\mbox{ subject to }\; &\Gamma(\Psi_j)=\bR_j \mbox{ for } j=0,1.
\end{aligned}
\end{equation}
\endgroup
Typically, the cost function $\costfunc$ would be selected to be symmetric in its arguments, in which case  $\Omtcov_\kappa$ would be symmetric as well, i.e., $\Omtcov_\kappa(\bR_0,\bR_1) = \Omtcov_\kappa(\bR_1,\bR_0)$ for any Toeplitz matrices  $\bR_0, \bR_1 \in \RM_+$. Although many possible choices of such cost functions exist, we will in most examples presented herein consider $\costfunc(\theta,\varphi) = \abs{e^{i\theta} - e^{i\varphi}}^2$, i.e., the cost function quantifies distances as the square of the distance between the corresponding points on the unit circle. As we show in Section \ref{sec:sos}, this particular choice of cost function allows for a sum-of-squares relaxation of the dual formulation of \eqref{eq:omt_toeplitz_kappa2}. This dual formulation is presented next.
\subsection{The dual formulation}
In order to study properties of the distance notion $\Omtcov_\kappa$, we consider the dual formulation of \eqref{eq:omt_toeplitz_kappa2}, where we assume that the cost function, $\costfunc$, is a continuous non-negative function on $\RT^2$. In order to address this problem, we first note that the adjoint operator%
\footnote{Strictly speaking, this is the pre-adjoint operator of $\Gamma$.} $\Gamma^*:\RM^{n}\to C_{\rm perio}(\RT)$ of the operator $\Gamma$ is
\begin{align*}
	\Gamma^*(\bR)(\theta) = \frac{1}{2\pi}\ba(\theta)^H\bR \ba(\theta)
\end{align*}
since
\begin{align*}
	\langle \Gamma(\Phi), \bR \rangle &= \langle \frac{1}{2\pi} \int_\RT \ba(\theta) \Phi(\theta)\ba(\theta)^Hd\theta, \bR \rangle 
	\\
	&= \int_\RT \frac{1}{2\pi}\Phi(\theta) \ba(\theta)^H\bR \ba(\theta)d\theta = \langle \Phi, \Gamma^*(\bR) \rangle
\end{align*}
where, in the first line, the inner product is the one associated with $\RM^n$, and the second line is the bilinear form with arguments $\Gamma^*(\RR)\in C_{\rm perio}(\RT)$ and $\Phi\in \ccM_+(\RT)$. With this result, we can derive an expression of the dual problem by considering the Lagrangian relaxation of~\eqref{eq:omt_toeplitz_kappa2}. The Lagrangian is given by
\begin{align*}
&\mathcal{L}_\kappa(M,\Psi_0,\Psi_1,\bLambda_0,\bLambda_1) = \int_{\RT^2} \costfunc(\theta,\varphi)M(\theta,\varphi)d\theta d\varphi \\
&\quad+ \langle \bR_0 \!-\! \Gamma(\Psi_0),\bLambda_0 \rangle +\langle \bR_1 \!-\! \Gamma(\Psi_1),\bLambda_1  \rangle \\
 &\quad+\kappa\norm{\int_{\RT}M(\theta,\varphi)d\varphi - \Psi_0}_1 + \kappa\norm{\int_{\RT}M(\theta,\varphi)d\theta - \Psi_1}_1  \\
 &= \langle \bLambda_0, \bR_0  \rangle + \langle \bLambda_1, \bR_1  \rangle +  \int_{\RT^2} \costfunc(\theta,\varphi)M(\theta,\varphi)d\theta d\varphi\\
 &\quad- \langle \Psi_0, \Gamma^*(\Lambda_0)  \rangle - \langle \Psi_1, \Gamma^*(\Lambda_1)  \rangle \\
 &\quad+\kappa\norm{\int_{\RT}M(\theta,\varphi)d\varphi - \Psi_0}_1 + \kappa\norm{\int_{\RT}M(\theta,\varphi)d\theta - \Psi_1}_1.
\end{align*}
Note that the Lagrange multiplier matrices $\bLambda_0$ and $\bLambda_1$ may be taken as Hermitian matrices, as $\bR_0 - \Gamma\left(\Psi_0\right)$ and $\bR_1 - \Gamma\left(\Psi_1\right)$ are Hermitian, and thus all inner products are real. Considering the infimum of $\mathcal{L}_\kappa$ with respect to $\Psi_0$ and $\Psi_1$, it may be noted that this is only finite if $\Gamma^*(\Lambda_0)(\theta) \leq \kappa$ and $\Gamma^*(\Lambda_1)(\varphi) \leq \kappa$, for all $\theta, \varphi \in \RT$. If this is satisfied, the Lagrangian is, for any fixed non-negative $M$, minimized by $\Psi_0$ and $\Psi_1$ given by
\begin{align*}
	\Psi_0(\theta) &= \indfunc{\Gamma^*(\Lambda_0)(\theta) \in [-\kappa,\kappa]} \int_{\RT}M(\theta,\varphi)d\varphi \\
	\Psi_1(\varphi) &= \indfunc{\Gamma^*(\Lambda_1)(\varphi) \in [-\kappa,\kappa]} \int_{\RT}M(\theta,\varphi)d\theta,
\end{align*}
where $\indfunc{\cdot}$ is the indicator function. Using this, and considering the infimum with respect to $M$, we arrive at
{
\begin{align*}
\inf_{\substack{M\in \ccM_+(\RT^2)\\ \Psi_0, \Psi_1 \in \ccM_+(\RT)}} \!\mathcal{L}_\kappa &=  \begin{cases}
		\langle \bLambda_0, \bR_0  \rangle \!+\! \langle \bLambda_1, \bR_1  \rangle & \!\text{if } \:(\bLambda_0,\bLambda_1) \in \Omega_\costfunc^\kappa \\
		-\infty & \!\text{otherwise}
	\end{cases}
\end{align*}
}
where
\begin{align*}
	\Omega_\costfunc^\kappa = \big\{ &\bLambda_0, \bLambda_1 \in \RM^{n}  \mid\\
	&\Gamma^*(\bLambda_0)(\theta)+ \Gamma^*(\bLambda_1)(\varphi) \leq \costfunc(\theta,\varphi), \\
	& \Gamma^*(\bLambda_0)(\theta) \leq \kappa ,\; \Gamma^*(\bLambda_1)(\varphi) \leq \kappa \; \mbox{ for all }  \theta,\varphi \in \RT \big\},
\end{align*}
where we have used that the cost function $\costfunc$ is non-negative. This yields the dual problem.
%
%
%
%%%%%% PROPOSITION: DUAL PROBLEM %%%%%%%
\begin{proposition}
Let the cost function $c$ be continuous and non-negative, and let $\kappa >0$. Then, the dual problem of \eqref{eq:omt_toeplitz_kappa2} is
\begin{align} \label{eq:dual_problem}
	\maxemphwrt[(\bLambda_0,\bLambda_1) \in \Omega_\costfunc^\kappa] & \quad \langle \bLambda_0, \bR_0  \rangle + \langle \bLambda_1, \bR_1  \rangle
\end{align}
where
\begin{align*}
	\Omega_\costfunc^\kappa = \big\{ &\bLambda_0, \bLambda_1 \in \RM^{n}  \mid\\
	&\Gamma^*(\bLambda_0)(\theta)+ \Gamma^*(\bLambda_1)(\varphi) \leq \costfunc(\theta,\varphi), \\
	& \Gamma^*(\bLambda_0)(\theta) \leq \kappa ,\; \Gamma^*(\bLambda_1)(\varphi) \leq \kappa \;\; \forall  \theta,\varphi \in \RT \big\}.
\end{align*}
\end{proposition}
%%%%%%%%%%%%%%%%%%%%%%%%%%%%%%%%%%%%%%%%%%%%%%
%
Since the primal and dual problems are convex and the set of feasible points $\Omega_\costfunc^\kappa$ has non-empty interior for any \mbox{$\kappa >  0$}, Slater's condition (see, e.g., \cite{BoydV04}) gives that strong duality holds and hence the duality gap between \eqref{eq:omt_toeplitz_kappa2} and \eqref{eq:dual_problem} is zero. This can be generalized to hold for lower semi-continuous cost functions analogous to \cite[proof of Theorem 4.1]{Edwards10_200}. 
Also, note the strong resemblance in form compared to the dual of \eqref{eq:Skappa}, which is given by \cite{GeorgiouKT09_57}
\begin{align*}
\maxwrt[\lambda_0, \lambda_1\in C_{\rm perio}(\RT)] \; & \int_\RT \lambda_0(\theta)\Phi_0(\theta)d\theta+\int_\RT \lambda_1(\varphi)\Phi_1(\varphi)d\varphi\\
\mbox{ subject to } \quad &  \lambda_0(\theta)\!+\!\lambda_1(\varphi)\le \costfunc(\theta, \varphi) \;\; \mbox{ for all } \theta, \varphi \in \RT\\
&\lambda_0(\theta)\le \kappa \,\qquad \qquad \qquad \mbox{ for all } \theta \in \RT\\
&\lambda_1(\varphi)\le \kappa \qquad \qquad \qquad \mbox{ for all } \varphi \in \RT.
\end{align*}
Similarly, for the case when the diagonals of $\bR_0$ and $\bR_1$ are required to be equal, the dual problem is
\begin{align} \label{eq:dual_problem_no_kappa}
	\maxwrt[(\bLambda_0,\bLambda_1) \in \Omega_\costfunc] & \quad \langle \bLambda_0, \bR_0  \rangle + \langle \bLambda_1, \bR_1  \rangle,
\end{align}
where
\begin{align*}
	\Omega_\costfunc = \big\{ &\bLambda_0, \bLambda_1 \in \RM^{n}  \mid\\
	&\Gamma^*(\bLambda_0)(\theta)+ \Gamma^*(\bLambda_1)(\varphi) \leq \costfunc(\theta,\varphi) \;\; \mbox{ for all }  \theta,\varphi \in \RT \big\}.
\end{align*}
This is the dual problem of \eqref{eq:omt_toeplitz}, where adding and subtracting mass is not allowed in the transport problem.

\subsection{Properties of distance notion $\Omtcov_\kappa$}
For the distance notion $\Omtcov_\kappa$ in \eqref{eq:omt_toeplitz_kappa}, the following proposition holds.
%
%%%%% PROPOSITION: SEMI-METRIC %%%%%%%
\begin{proposition}\label{prop:semi-metric}
Let $\kappa > 0$ and let the cost function $\costf$ be a continuous function and a semi-metric on $\RT$. Then, the distance notion $\Omtcov_\kappa$ in \eqref{eq:omt_toeplitz_kappa} is a semi-metric on the set of positive semi-definite Toeplitz matrices.
\end{proposition}
%%%%%%%%%%%%%%%%%%%%%%%%%%%%%%%%%%%%%%%%%%%%%
%
%%% PROOF %%%%
\begin{proof}
See appendix.
\end{proof}
%%%%%%%%%%%
%
The implication of the semi-metric property is that $\Omtcov_\kappa$ may indeed be used to quantify distances between covariance matrices, or, stochastic processes. We may also state the following proposition.
%
%%%%% PROPOSITION: CONTRACTIVENESS OF THE DISTANCE NOTION %%%%%%%
\begin{proposition}\label{prop:contractive_prop}
Let the cost function $\costfunc$ be continuous and non-negative and such that $\costfunc(\theta,\theta) = 0$, for all $\theta \in \RT$. Then, the distance measure $\Omtcov_\kappa$ defined in \eqref{eq:omt_toeplitz_kappa} is contractive with respect to additive noise.
If $\costfunc$ is also shift-invariant, i.e., \mbox{$\costfunc(\theta-\phi,\varphi-\phi) = \costfunc(\theta,\varphi)$}, for all $\theta,\varphi,\phi \in \RT$, then $\Omtcov_\kappa$ is also contractive with respect to multiplicative noise whose covariance matrix has diagonal elements smaller than or equal to unity.
\end{proposition}
\begin{proof}
See appendix.
\end{proof}
%%%%%%%%%%%%%%%%%%%%%%%%%%%%%%%%%%%%%%%%%%%%%
%
The statement of this proposition is that when two stochastic processes become contaminated by noise, the distance notion $\Omtcov_\kappa$ decreases and hence the processes become harder to distinguish. Intuitively, this is a desirable property of $\Omtcov_\kappa$, as additive or multiplicative noise should indeed impede ones ability to discriminate between two processes. Proposition~\ref{prop:contractive_prop} may be proven by utilizing results in \cite{GeorgiouKT09_57}. However, in the interest of making the exposition self-contained, we  provide a direct proof in the appendix based on the dual formulation in \eqref{eq:dual_problem}.
It may be noted that the assumptions made in Proposition~\ref{prop:contractive_prop} regarding the cost function $\costfunc$ are quite mild, allowing for a large class of potential cost functions. In particular, \mbox{$\costfunc(\theta,\varphi) = \abs{e^{i\theta}-e^{i\varphi}}^2$} satisfies the conditions of Proposition~\ref{prop:contractive_prop} for both additive and multiplicative noise.
\section{Interpolation, extrapolation, and tracking}\label{sec:applications}
The formulation in \eqref{eq:omt_toeplitz} does not only define a notion of distance between two Toeplitz covariance matrices, $\bR_0$ and $\bR_1$; it also provides a means of forming interpolating matrices, i.e., defining intermediate covariance matrices $\bR_\tau$, for $\tau \in (0,1)$. In order to define the interpolating matrices, we again utilize the spectral representation of positive semi-definite Toeplitz matrices. To this end, we note that, given an optimal transport plan, $M$, found as the functional minimizing \eqref{eq:omt_toeplitz}, one may define spectra intermediate to the marginals $\int_\RT M(\theta,\varphi)d\varphi$ and $\int_\RT M(\theta,\varphi)d\theta$ by linearly shifting the frequency locations of the spectral mass, as dictated by $M$. That is, any mass transferred from $\phi$ to $\phi+\varphi$ is defined to, at $\tau \in [0,1]$, be located at frequency $\phi+\tau \varphi$. Using this, the intermediate spectrum is given by
\begin{align}
	\Phi_\tau^M(\theta) &\triangleq \int_{\RT^2}\delta_\theta(\{\phi + \tau\varphi\}_{{\rm mod}\, \RT}) M(\phi,\phi+\varphi) d\phi d\varphi \nonumber\\
	&= \int_\RT M(\theta-\tau\varphi, \theta+(1-\tau)\varphi) d\varphi.
\end{align}
Here, $\delta_\theta$ is the Dirac delta function localized at $\theta$, i.e., \mbox{$\delta_\theta(\phi) \triangleq \delta(\phi-\theta)$} and the integrands are extended periodically with period $2\pi$ outside the domain of integration. Also, we denote by $\{x\}_{{\rm mod}\, \RT}$ the value in $\RT$ that is congruent with $x$ modulo $2\pi$. Based on this definition of the interpolating spectrum $\Phi_\tau^M$, the corresponding interpolating covariance matrix, $\bR_\tau$ is defined as the unique Toeplitz matrix consistent with this spectrum, i.e.,  
\begin{align}\label{eq:geodesic}
	 & \bR_\tau\triangleq \Gamma(\Phi_\tau^M)\\
	&= \frac{1}{2\pi} \!\int_\RT\ba(\theta)\left(\int_\RT M(\theta-\tau\varphi, \theta+(1-\tau)\varphi) d\varphi\right)\ba(\theta)^Hd\theta \nonumber \\
	&=\frac{1}{2\pi} \!\int_{\RT^2}\!\ba(\{\phi\!+\!\tau\varphi\}_{{\rm mod}\, \RT})M(\phi,\phi\!+\!\varphi) \ba(\{\phi\!+\!\tau\varphi\}_{{\rm mod}\, \RT})^Hd\phi d\varphi \nonumber
\end{align}
for $\tau \in [0,1]$. To simplify the following exposition, let $\ccI_\tau(M)\triangleq\Gamma(\Phi_\tau^{M})$ denote the linear operator in \eqref{eq:geodesic} that maps a transport plan to an interpolating covariance matrix, i.e., $\bR_\tau = \ccI_\tau(M)$. It may be noted from \eqref{eq:geodesic} that $\ccI_\tau(M)$ is well-defined also for $\tau \neq [0,1]$, i.e., the formulation allows also for extrapolation. The following proposition follows directly from the definition in \eqref{eq:geodesic}.
\begin{proposition}\label{prp:properties}
For any $\tau\in \RR$, the following basic properties hold for $\bR_\tau = \ccI_\tau(M)$\emph{:}
\begin{enumerate}
\setlength{\itemsep}{1pt}
\setlength{\parskip}{1pt}
\item[a)] If $\bR_0$ and $\bR_1$ have the same diagonal, then it is also the diagonal of $\bR_\tau$.
\item[b)] The matrix $\bR_\tau$ is a Toeplitz matrix.
\item[c)] The matrix $\bR_\tau$ is positive semi-definite.
\end{enumerate}
\end{proposition}
Due to these properties, the proposed method offers a way of interpolating the covariances of, e.g., slowly varying time series as the interpolant $\bR_\tau$ allows for modeling linear changes in the spectrum of the process.
\begin{remark}
The interpolation approach generalizes trivially to the formulation in \eqref{eq:omt_toeplitz_kappa2} between the covariances $\Gamma(\Phi_0^M)$ and $\Gamma(\Phi_1^M)$ under the assumption that $\Psi_0$ and $\Psi_1$ are perturbations of $\Phi_0^M$ and $\Phi_1^M$, respectively. In order to define an interpolation and extrapolation procedure from $\bR_0$ to $\bR_1$ where there is a cost $\kappa$ for adding and subtracting mass, one may, along with the interpolation path $\ccI_\tau(M)$, linearly add the part corresponding to the added and subtracted mass, i.e., 
\begin{align}
\!\bR_\tau\!=
\ccI_\tau(M)+(1-\tau)\Gamma(\Psi_0-\Phi^M_0) +\tau\Gamma(\Psi_1-\Phi^M_1).
\end{align}
Note that in this scenario, positive semi-definiteness cannot be guaranteed for the extrapolation case.
\end{remark}

%
%

%%%%%% SUB-SECTION: OTHER METHODS %%%%%%%%
\subsection{Comparison with other methods}
The properties in Proposition \ref{prp:properties} distinguish the proposed interpolant $\bR_\tau$ from other proposed matrix geodesics. 
As an example, consider the interpolant induced by the Euclidean metric, i.e., the distance between two covariance matrices $\bR_0$ and $\bR_1$ is defined as $\norm{\bR_0 - \bR_1}_F$. This yields interpolants that are formed as convex combinations of $\bR_0$ and $\bR_1$, i.e., $\bR^{\rm conv}_{\tau}=\tau \bR_0 + (1-\tau)\bR_1$, for $\tau \in [0,1]$. This preserves the Toeplitz structure, as well as the diagonal of the end-point matrices and the positive semi-definiteness. However, from a spectral representation point of view, the convex combination gives rise to fade-in fade-out effects, i.e., only spectral modes directly related to $\bR_0$ and $\bR_1$ can be represented, and there can be no shift in the location of these modes (see also Example~\ref{ex:DOA} and Figure~\ref{fig:doa_interp_conv}).
Other more sophisticated options include, e.g., the geodesic with respect to g-convexity \cite{Wiesel12_60}
\begin{align} \label{eq:R_tilde}
	\tilde{\bR}_\tau = \bR_0^{1/2}\left(\bR_0^{-1/2}\bR_1 \bR_0^{-1/2}  \right)^\tau \bR_0^{1/2}
\end{align}
and the geodesic in \cite{KnottS84_43,NingJG13_20}, which builds on optimal mass transport of Gaussian distributions and can be expressed as
\begin{align} \label{eq:R_breve}
	\breve{\bR}_\tau \!= \left((1-\tau)\bR_0^{1/2} \!+ \!\tau \bR_1^{1/2}\bU \right)\! \left((1-\tau)\bR_0^{1/2} \!+\! \tau \bR_1^{1/2}\bU \right)^H
\end{align}
where
\[
\bU = \bR_1^{-1/2}\bR_0^{-1/2}\left(\bR_0^{1/2}\bR_1\bR_0^{1/2}\right)^{1/2}.
\]
One may also perform interpolation using geodesics induced by the log-Euclidean metric (see, e.g., \cite{JayasumanaHSLH13_cvpr}), i.e., where distances are defined as $\norm{\log\left(\bR_0\right) - \log\left(\bR_1\right)}_F$, with $\log\left( \cdot\right)$ here denoting the matrix logarithm. For this case, the geodesic is given by
\begin{align} \label{eq:R_log_euclid}
	\bR^{\text{log-Euclid}}_\tau = \exp\left( (1-\tau)\log\left( \bR_0 \right) + \tau \log\left(\bR_1 \right)  \right),
\end{align}
where $\exp\left( \cdot \right)$ denotes the matrix exponential. It may here be noted, that although the three geodesics in \eqref{eq:R_tilde}, \eqref{eq:R_breve}, and \eqref{eq:R_log_euclid} preserve positive definiteness, they are not defined for singular matrices due to the use of matrix inverses and matrix logarithms. Also, for general Toeplitz covariance matrices, these geodesics preserve neither the Toeplitz structure nor the diagonal of the end-point matrices.
Further, as noted above, the three properties in Proposition \ref{prp:properties} hold for any $\tau \in  \RR$ for the proposed approach, and thus directly allows for extrapolating using \eqref{eq:geodesic}. In contrast, it may be noted that for the linear combination $\bR^{\rm conv}_{\tau}$ there are no guarantees that the resulting matrix is positive semi-definite if $\tau\notin [0,1]$. Also, note that the alternative geodesics in \eqref{eq:R_tilde}, \eqref{eq:R_breve}, and \eqref{eq:R_log_euclid} do not naturally generalize to extrapolation.
%
%
%%%%%%%%% SUB-SECTION: BARY CENTER %%%%%%%
\subsection{Tracking of slowly varying processes}
The proposed interpolant $\bR_\tau = \ccI_\tau(M)$ may also be readily used for tracking slowly varying stochastic process. As noted above, $\ccI_\tau(M)$ allows for the modeling of slow, i.e., locally linear, shifts in the location of spectral power. Building on this property, we can extend the optimal transport problem in \eqref{eq:omt_toeplitz}, in order to fit a covariance path $\bR_\tau$ to a sequence of $J$ covariance matrix estimates, $\hat{\bR}_{\tau_j}$, for $j = 1,\ldots,J$. As $\bR_\tau$ is unambiguously determined from a transport plan $M$ via $\ccI_\tau(M)$, this tracking problem may be formulated as the convex optimization problem
\begin{align}\label{eq:omt_tracking}
\minwrt[M\in \ccM_+(\RT^2)] & \int_{\RT^2} \!c(\theta,\varphi)
M(\theta,\varphi)d\theta d\varphi\!+\!\lambda\sum_{j=1}^J\norm{\ccI_{\tau_j}(M)\! -\!\hat{\bR}_{\tau_j}}_F^2
\end{align}
where $\lambda > 0$ is a user-specified regularization parameter. As may be noted from \eqref{eq:omt_tracking}, the optimal transport plan $M$ is here determined as the one that minimizes not only the transport cost, but also takes into account the deviations of the interpolant $\bR_{\tau_j} = \ccI_{\tau_j}(M)$ from the available covariance matrix estimates $\hat{\bR}_{\tau_j}$. The behavior of this construction is illustrated in the numerical section.

%%%%%%%%% SUB-SECTION: BARY CENTER %%%%%%%
\subsection{Clustering: Barycenter computation} \label{sec:barycenter}
As a further example, we will see that the barycenter with respect to the distance notion $\Omtcov_\kappa$ may be formulated as a convex optimization problem. This might be desirable in clustering or classification applications, where one is interested in either identifying classes of signals or processes based on their covariance matrices or associate a given covariance matrix with such a signal class. Considering the case of clustering,  assume that $L$ covariance matrices, $\bR_\ell$, $\ell = 1,\ldots, L$, are available. Then, we may define their barycenter via $\Omtcov_\kappa$ according to
\begin{align} \label{eq:omt_barycenter}
	\bR_{\text{bary}} = \argminwrt[\bR\in \RM_+^{n}] \sum_{\ell=1}^L \Omtcov_\kappa(\bR, \bR_\ell),
\end{align}
i.e., as the covariance matrix that minimizes the sum of $\Omtcov_\kappa$ for the set of covariance matrices $\bR_\ell$. Explicitly, $\bR_{\text{bary}}$ solves the convex optimization problem
\begin{equation} \label{eq:omt_barycenter_explicit}
\begin{aligned}
	\minwrt[\substack{\bR\in \RM_+^{n}\\ M_\ell\in \ccM_+(\RT^2) \\ \Psi_\ell\in \ccM_+(\RT)\\\Phi_\ell\in \ccM_+(\RT)}]\quad &  \sum_{\ell=1}^L \int_{\RT^2} \costfunc(\theta,\varphi) M_\ell(\theta,\varphi) d\theta d\varphi \\
	&+\kappa\sum_{\ell=1}^L\norm{\int_{\RT} M_\ell(\theta,\phi)d\theta-\Phi_\ell}_1 \\
	&+\kappa\sum_{\ell=1}^L\norm{\int_{\RT} M_\ell(\theta,\phi)d\phi-\Psi_\ell}_1 \\
	\text{subject to}\quad & \quad\Gamma(\Phi_\ell) = \bR \\
	&\quad\Gamma(\Psi_\ell) = \bR_\ell ,\; \ell = 1,\ldots,L.
\end{aligned}
\end{equation}
This formulation allows for using, e.g., K-means clustering (see, e.g., \cite{DellerHP99}) in order to identify classes of covariance matrices, as well as classify a given covariance matrix according to these classes. Classification of a covariance matrix $\bR$ according to classes defined by a set of barycenters $\bR_{\text{bary}}^{(j)}$, $j = 1,\ldots, J,$ may then be formulated as
\begin{align} \label{eq:omt_barycenter_classify}
	\argminwrt[j \in \{1,\ldots,J  \}] \Omtcov_\kappa(\bR,\bR_{\text{bary}}^{(j)}).
\end{align}
In Section~\ref{ssec:Clustering}, we present a simple illustration of this potential application, considering unsupervised clustering of phonemes.
%

%%%%%%%%% Figure: Basic example %%%%%%%%%%%%%%%%%
\begin{figure}[t!]
        \centering
        %\vspace{-2.5 mm}
            \includegraphics[width=.48\textwidth]{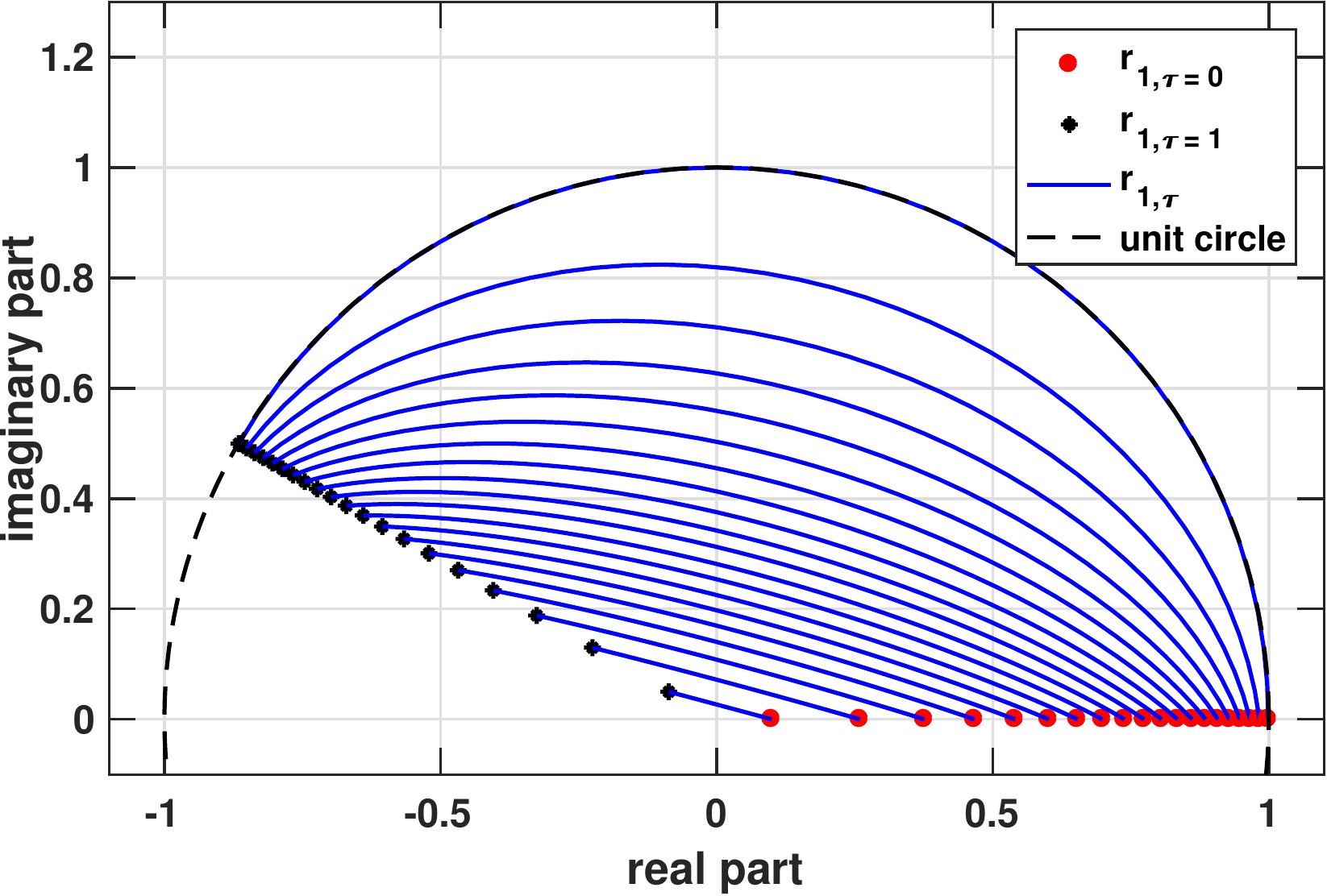}
           \caption{Path of the off-diagonal element of the covariance matrix $\bR_\tau$ for the case $n = 2$, when choosing $\bR_0$ and $\bR_1$ according to \eqref{eq:basic_example_R0_R1}.}
            \label{fig:basic_example}
\vspace{-2mm}\end{figure}
%%%%%%%%%%%%%%%%%%%%%%%%%%%%%%%%%
%
%
%
%
%
%
\section{Sum-of-squares relaxation} \label{sec:sos}
In order to solve the dual problem in \eqref{eq:dual_problem} in practice, it has to be implemented as a finite dimensional problem, e.g., by gridding the space $\RT^2$ and thereby approximating the set $\Omega_\costfunc^\kappa$ using a finite number of constraints. However, for the special case of $\costfunc(\theta,\varphi) = \abs{e^{i\theta}-e^{i\varphi}}^{p}$, for $p \in \RN$, this can also be done using a sum-of-squares (SOS) relaxation (see, e.g., \cite{Dumitrescu06_54}) of \eqref{eq:dual_problem}. For simplicity of notation, we present here the case with $\costfunc(\theta,\varphi) = \abs{e^{i\theta}-e^{i\varphi}}^{2}$ and without $\kappa$ ($\kappa=\infty$), i.e., the dual of \eqref{eq:omt_toeplitz} as formulated in \eqref{eq:dual_problem_no_kappa}. To this end, identify
\begin{align*}
	z = e^{i\theta},\; w = e^{i\varphi},
\end{align*}
which allows us to write
\begin{align*}
	\costfunc(\theta,\varphi) &= 2 - zw^{-1} - z^{-1}w \\
	\Gamma^*(\bLambda_0)(\theta) &= \frac{1}{2\pi}\sum_{k=1}^{n}\sum_{\ell=1}^{n} \left[ \bLambda_0  \right]_{k,\ell} z^{\ell-k} \\
	\Gamma^*(\bLambda_1)(\varphi) &=\frac{1}{2\pi}\sum_{k=1}^{n}\sum_{\ell=1}^{n} \left[ \bLambda_1  \right]_{k,\ell} w^{\ell-k}.
\end{align*}
Thus, the set of constraints defining the feasible set $\Omega_\costfunc$ is given by $\Gamma^*(\bLambda_0)(\theta)+\Gamma^*(\bLambda_1)(\varphi) \le\costfunc(\theta,\varphi)$, for all $\theta,\varphi \in \RT$, or, equivalently,
\begingroup
\small
\begin{align}\label{eq:Ppoly}
2\! - \!zw^{-1} \!- \!z^{-1}w \!-\!\frac{1}{2\pi}\!\sum_{k,\ell=1}^{n}\!\left[ \bLambda_0  \right]_{k,\ell}\!z^{\ell-k} \!-\!\frac{1}{2\pi}\!  \sum_{k,\ell=1}^{n}\!\left[ \bLambda_1  \right]_{k,\ell}\!w^{\ell-k} \geq 0
\end{align}
\endgroup
Note that in the two-dimensional trigonometric polynomial \eqref{eq:Ppoly}, the coefficient for $z^{-k_1}w^{-k_2}$ is equal to
\begin{equation*}
\begin{array}{ll}
		2-\frac{1}{2\pi}\text{diag}\left(\bLambda_0+\bLambda_1\right)^T\mathbf{1}&\text{ for } k_1=k_2 = 0 \\
		-\frac{1}{2\pi}\text{diag}\left( \bLambda_0 , k_1 \right)^T\mathbf{1}& \text{ for } k_1\in \RZ_{n \setminus 0} \text{ and } k_2= 0 \\
		-\frac{1}{2\pi}\text{diag}\left( \bLambda_1 , k_2 \right)^T\mathbf{1}& \text{ for } k_1= 0 \text{ and } k_2 \in \RZ_{n \setminus 0} \\
		-1 & \text{ for } k_1 = 1 \text{ and } k_2 = -1 \\
		-1 & \text{ for } k_1 = -1 \text{ and } k_2 = 1 \\
		0 & \text{ otherwise},
\end{array}
\end{equation*}
where $\RZ_{n \setminus 0} = \left\{ k \in \RZ \mid \abs{k}< n, k \neq 0  \right\}$. Here, $\text{diag}\left(\bX, k \right)$ denotes the column vector containing the elements on the $k$th super-diagonal of the matrix $\bX$, if $k>0$, and the elements on the $k$th sub-diagonal, if $k<0$, with $\mathbf{1}$ denoting a column vector of ones of appropriate dimension.

In order to formulate a computationally feasible optimization problem, we remove the non-negativity constraint for the two-dimensional polynomial in \eqref{eq:Ppoly} and instead impose that it should have a sum-of-squares representation \cite{Dumitrescu06_54},\cite{Dumitrescu05_53}. In particular, we impose that the polynomial in \eqref{eq:Ppoly} should be of the form 
\begin{align}\label{eq:pSOS}
	P(z,w) = \left(\bz^{-1}\right)^T \bQ\bz
\end{align}
where $\bQ\in \RM^{m^2}$ is positive semi-definite,  
\begin{align*}
 	\bz &= \begin{bmatrix} 1 & w & \cdots & w^{m-1}  \end{bmatrix}^T \otimes  \begin{bmatrix} 1 & z & \cdots & z^{m-1}  \end{bmatrix}^T \\
	\bz^{-1} \!&=\! \begin{bmatrix} 1 & w^{-1} & \cdots & w^{-m+1}  \end{bmatrix}^T \!\otimes \!\begin{bmatrix}1&z^{-1} & \cdots &z^{-m+1}\end{bmatrix}^T
\end{align*}
and $\otimes$ is the Kronecker product.
Note that any polynomial on this form is non-negative by definition, and, furthermore, for any non-negative polynomal $P^*$, there is a sequence of polynomials, $P_m$, on the form \eqref{eq:pSOS} such that \mbox{$\norm{P_m- P^*}_\infty \to 0$} as $m\to\infty$.

Next, note that the coefficients of $P$ are associated to the elements of $\bQ$ in \eqref{eq:pSOS} according to 
\begin{align*}
	p_{k_1,k_2} = \tr\left(\bT_{k_1,k_2} \bQ  \right), \quad \mbox{for } -m+1\leq k_1,k_2 \leq m-1
\end{align*}
where 
\begin{align*}
	\bT_{k_1,k_2} = \bT_{k_2} \otimes \bT_{k_1}
\end{align*}
and $\bT_k$ is the matrix with ones on the $k$th diagonal and zeros elsewhere.
%
%
%
%%%%% Figure: DOA interpolation,  proposed method %%%%%%%%%%%%
\begin{figure}[t!]
        \centering
        %\vspace{-2.5 mm}
            \includegraphics[width=.48\textwidth]{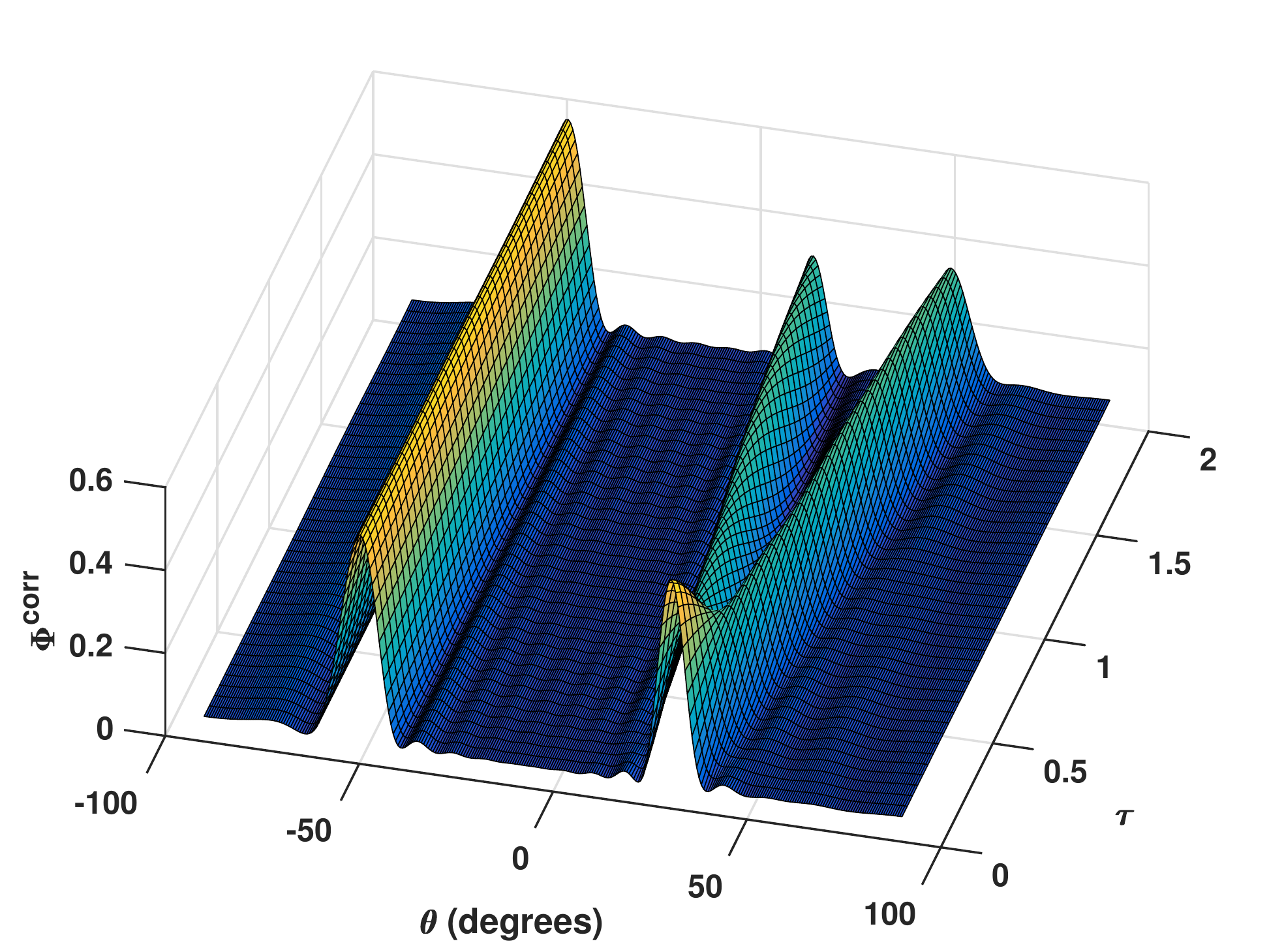}
           \caption{Interpolated spatial spectrum estimated as $\ba(\theta)^H\bR_{\tau}\ba(\theta)$, where $\bR_{\tau}$ is obtained by solving \eqref{eq:omt_toeplitz2}.
           }
            \label{fig:doa_interp_gt}
\vspace{-2mm}\end{figure}
%%%%%%%%%%%%%%%%%%%%%%%%%%%%%%%%%
%
%
%
%
Putting these facts together by requiring that the polynomial \eqref{eq:Ppoly} can be written as a sum-of-square \eqref{eq:pSOS}, we approximate \eqref{eq:dual_problem_no_kappa} by the semi-definite program (SDP)
\begingroup
\begin{align}
	&\maxwrt[\substack{\bLambda_0,\bLambda_1 \in \RM^{n} \\\bQ\in\RM^{m^2}}] && \langle \bLambda_0, \bR_0  \rangle + \langle \bLambda_1,\bR_1  \rangle \label{eq:SDPoptimization}\\
	& \text{subject to} &&  \bQ \succeq 0\nonumber \\
	&&&  \tr\left( \bQ  \right) = 2-\frac{1}{2\pi}\text{diag}\left(\bLambda_0 + \bLambda_1\right)^T\mathbf{1} \nonumber\\
	&&& \tr\left(\bT_{-1,1} \bQ  \right) = -1 \nonumber\\
	&&& \tr\left(\bT_{1,-1} \bQ  \right) = -1 \nonumber\\
	&&& \tr\left(\bT_{k_1,0} \bQ  \right) = -\frac{1}{2\pi}\text{diag}\left( \bLambda_0 , k_1 \right)^T \mathbf{1} \:,\; k_1 \in \RZ_{n \setminus 0} \nonumber\\
	&&& \tr\left(\bT_{0,k_2} \bQ  \right) = -\frac{1}{2\pi}\text{diag}\left( \bLambda_1 , k_2 \right)^T \mathbf{1} \:,\; k_2 \in \RZ_{n \setminus 0} \nonumber\\
	&&& \tr\left(\bT_{k_1,k_2} \bQ  \right) = 0 \:,\; \text{if } k_1k_2 = 1 \nonumber\\
	&&&\hspace{3.2cm}\text{ or } 1<\abs{k_1k_2}\leq (m-1)^2. \nonumber
\end{align}
\endgroup
It is worth noting that the $P$ defined by the optimal $\bQ$ in \eqref{eq:pSOS} is non-negative on the unit torus. Thus, for the solution of the SDP problem in \eqref{eq:SDPoptimization}, there will be a corresponding feasible solution of \eqref{eq:dual_problem_no_kappa}. Therefore, any solution to \eqref{eq:SDPoptimization} will give a lower bound for the optimal objective value of \eqref{eq:dual_problem_no_kappa}. 
 However, any non-negative polynomial may be arbitrarily well approximated by $P$ of the form \eqref{eq:pSOS} by a suitable choice the degree $m-1$, and thus the maximal objective value of \eqref{eq:SDPoptimization} will converge to the maximal objective value of \eqref{eq:dual_problem_no_kappa} as $m$ grows.
By comparison, directly discretizing $\RT$, thereby  approximating $\Omega_\costfunc$ using a finite number of constraints, yields an upper bound for the maximal objective value of \eqref{eq:dual_problem_no_kappa}. In this case, the maximal objective value will also converge to that of \eqref{eq:dual_problem_no_kappa}, this time from above, as the spacing of the discretization of the grid becomes finer.
%
%
%
%
%%%%%%% NUMERICAL EXAMPLES %%%%%%%
\section{Numerical examples}\label{sec:num}
In this section, we present some numerical examples illustrating different aspects and application areas of the proposed distance notion $\Omtcov_\kappa$, as well as the interpolant $\bR_\tau = \ccI_\tau(M)$. Throughout these examples, we will use the cost function $\costfunc(\theta,\varphi) = \abs{e^{i\theta}-e^{i\varphi}}^2$. It should be stressed that many other choices of cost functions are possible, allowing for flexibility in modeling the specific scenario one is interested in. Apart from specific modeling aspects, one may preferably pick cost functions satisfying the assumptions in Propositions~\ref{prop:semi-metric} and \ref{prop:contractive_prop}.
\subsection{Trajectory example}
In order to illustrate the behavior of the proposed interpolation method, we consider a simple scenario with covariance matrices of size $n = 2$. Consider covariance matrices of the form
\begin{align} \label{eq:basic_example_R0_R1}
	\bR_0 = \begin{bmatrix} 1 & r \\ r & 1 \end{bmatrix},\; \bR_1 = \begin{bmatrix} 1 & re^{i\frac{5\pi}{6}} \\ r e^{-i\frac{5\pi}{6}}& 1 \end{bmatrix}
\end{align}
where $r \in [-1,1]$. Thus, considering interpolating paths $\bR_\tau$, these will be on the form
\begin{align}
	\bR_\tau = \begin{bmatrix} 1 & r_{1,\tau} \\ \bar{r}_{1,\tau} & 1 \end{bmatrix}.
\end{align}
Figure~\ref{fig:basic_example} displays the real and imaginary part of $r_{1,\tau}$ for $\tau \in [0,1]$ when varying the magnitude $r$ of the off-diagonal element of $\bR_0$ and $\bR_1$ between 0 and 1. As can be seen, the trajectories of $r_{1,\tau}$ approximately correspond to convex combinations $(1-\tau)r + \tau re^{i\frac{5\pi}{6}}$ when $r$ is close to zero, whereas they are considerably curved for $r$ closer to 1. It may be noted that for the singular case, i.e., $r = 1$, the trajectory of $r_{1,\tau}$ coincides with the unit circle. Thus, we see that for covariance matrices that have consistent spectra that are essentially flat, the interpolant $\bR_\tau$ will be approximately equal to the convex combination $(1-\tau)\bR_0 + \tau\bR_1$. At the other extreme, the interpolant corresponding to covariance matrices $\bR_0$ and $\bR_1$ that have consistent spectra that are close to being singular will also have almost singular consistent spectra.
%
%%%%%% Figure: DOA interpolation, convex combination of R0, R1 %%%%%%%%%%%%
\begin{figure}[t!]
        \centering
        %\vspace{-2.5 mm}
            \includegraphics[width=.48\textwidth]{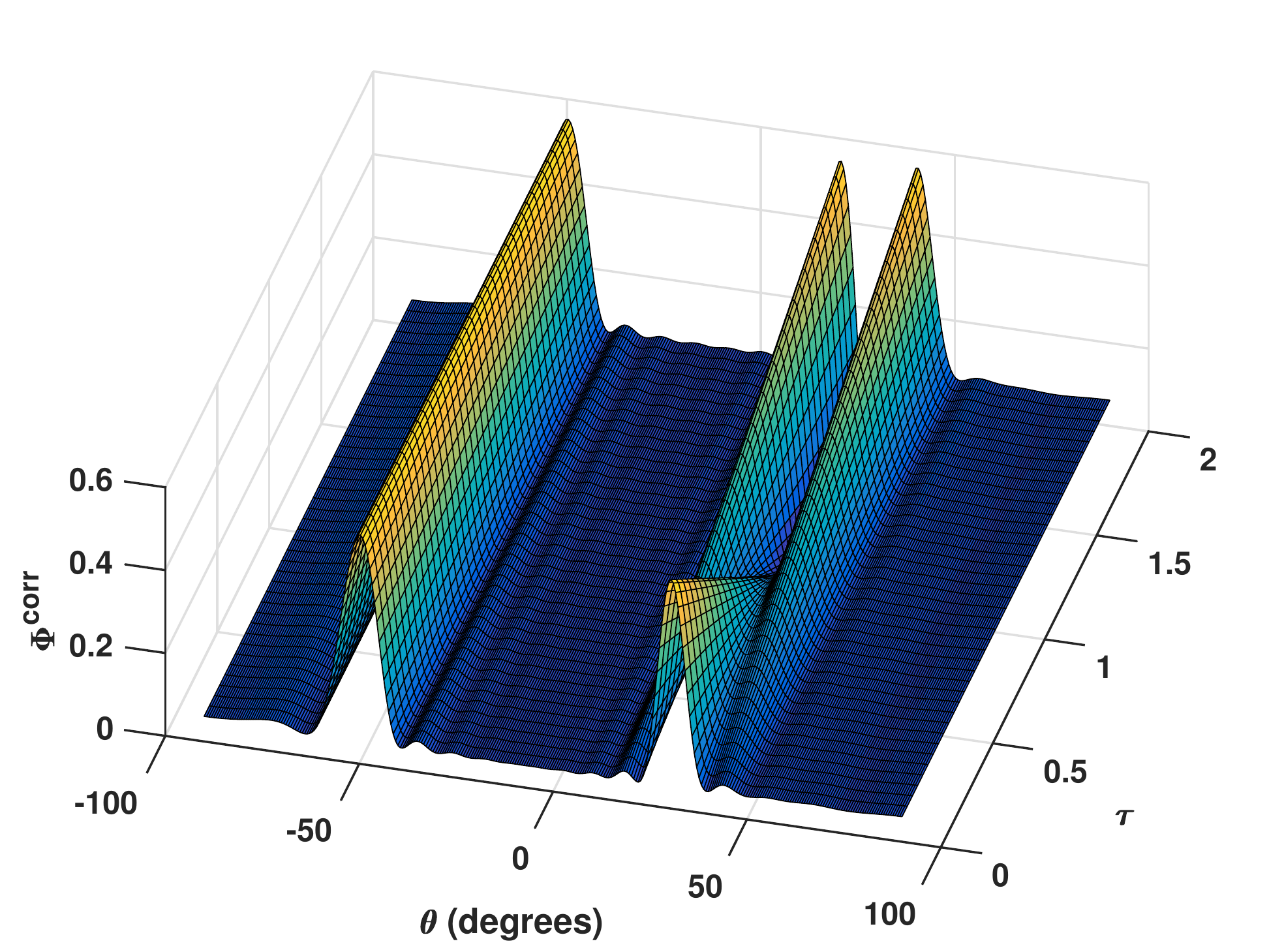}
           \caption{Interpolated spatial spectrum estimated as $\ba(\theta)^H\bR^{\rm conv}_\tau\ba(\theta)$, where $\bR_\tau^{\rm conv}$ is the linear combination of $\bR_0$ and $\bR_1$.%obtained by solving \eqref{eq:omt_toeplitz} with estimated covariances $\hat{\bR}^{(0)}$ and $\hat{\bR}^{(1)}$.
           }
            \label{fig:doa_interp_conv}
\vspace{-2mm}\end{figure}
%%%%%%%%%%%%%%%%%%%%%%%%%%%%%%%%%
%
%
%
%
%
%%%%%% DOA %%%%%%
\subsection{Interpolation and extrapolation for DOA}\label{ex:DOA}
Next, we illustrate the proposed methods ability to produce interpolants $\bR_\tau$ that are consistent with locally linear changes in the frequency location of spectral power. This is illustrated using a direction-of-arrival (DOA) estimation problem.
Consider a uniform linear array (ULA) with $15$ sensors with half-wavelength sensor spacing and a scenario where two covariance matrices
\begin{align*}
	\bR_0 &=\frac{1}{2}\sum_{\ell=1}^2\ba(\theta_\ell^{(0)})\ba(\theta_\ell^{(0)})^H + \sigma^2 \bI \\
	\bR_1 &=\frac{1}{2}\ba(\theta_1^{(1)})\ba(\theta_1^{(1)})^H+\frac{1}{4}\sum_{\ell=2}^3\ba(\theta_\ell^{(1)})\ba(\theta_\ell^{(1)})^H + \sigma^2 \bI
\end{align*}
are available\footnote{Note that $\theta$ in this example denotes spatial frequency. For simplicity, we retain the notation $\ba(\theta)$ also for this case.}. Here $\theta_1^{(0)} = \theta_1^{(1)} =-50\si{\degree}$, $\theta_2^{(0)} = 30\si{\degree}$, $\theta_2^{(1)} = 20\si{\degree}$, and $\theta_3^{(1)} = 40\si{\degree}$, and $\sigma^2 =0.05$. Such a scenario may be interpreted as a target at $\theta_2^{(0)}$ splitting up into two targets at $\theta_2^{(1)} $ and $\theta_3^{(1)}$ as time progresses, whereas the target at $\theta_1^{(0)}$ stays put.
%
%
%
%%%%% Figure: DOA interpolation, \tilde{R} %%%%%%%%%%%%
\begin{figure}[t!]
        \centering
        %\vspace{-2.5 mm}
            \includegraphics[width=.48\textwidth]{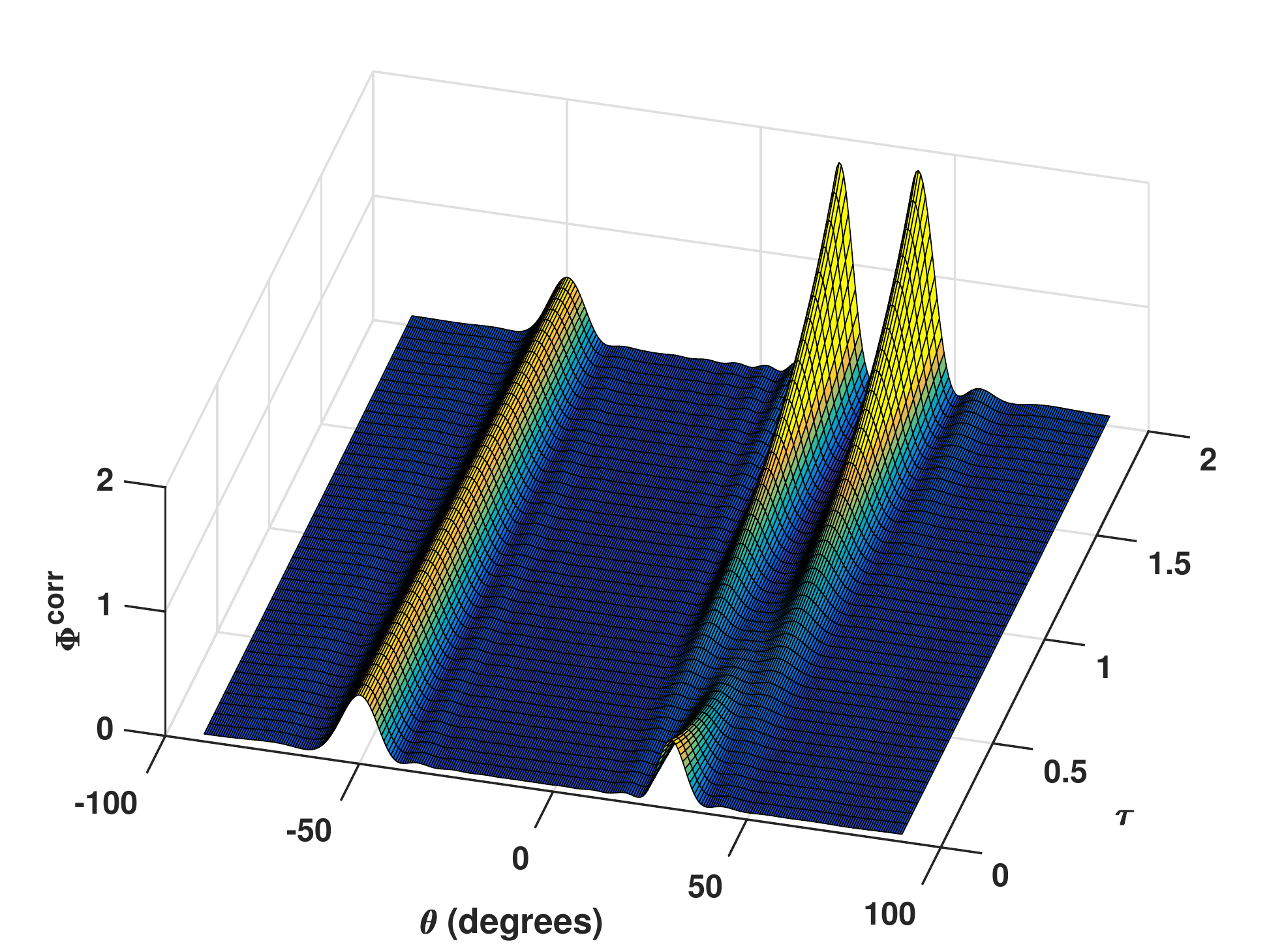}
           \caption{Interpolated spatial spectrum estimated as $\ba(\theta)^H\tilde{\bR}_\tau\ba(\theta)$, where $\tilde{\bR}_\tau$ is obtained from \eqref{eq:R_tilde}.
           }
            \label{fig:doa_interp_tilde}
\vspace{-2mm}\end{figure}
%%%%%%%%%%%%%%%%%%%%%%%%%%%%%%%%%
%
%
%
%
%
We use the proposed method in \eqref{eq:omt_toeplitz2}, as $\bR_0$ and $\bR_1$ have the same diagonal, in order to find the optimal transport map $M$. Then, using \eqref{eq:geodesic}, we compute covariance matrices $\bR_\tau$, for $\tau \in [0,2]$, i.e., we both interpolate on $\tau \in [0,1]$ and extrapolate on $\tau \in (1,2]$. This is then compared to the basic interpolant $\bR^{\rm conv}_{\tau}$ based on convex and linear combinations of $\bR_0$ and $\bR_1$, as well as the more sophisticated covariance matrix geodesics $\tilde{\bR}_\tau$ and $\breve{\bR}_\tau$, as defined in \eqref{eq:R_tilde} and \eqref{eq:R_breve}, respectively. For these four cases, we then estimate the corresponding inter- and extrapolated spectra using the correlogram, i.e., as
\begin{align} \label{eq:corr_gram}
	\Phi^{\text{corr}}(\bX,\theta)  = \ba(\theta)^H\bX\ba(\theta)
\end{align}
where $\bX$ is substituted for the four different covariance interpolants. The results for the proposed interpolant $\bR_\tau$ are shown in Figure~\ref{fig:doa_interp_gt}. As can be seen, $\bR_\tau$ indeed models a scenario where one of the targets has a constant location, whereas the second target splits upp into two smaller targets. Note also that $\bR_\tau$ implies that the smaller target continue linearly with respect to the look-angle, $\theta$, also for the extrapolation case, i.e., for $\tau>1$.
In contrast, the convex (linear for $\tau>1$) combination $\bR^{\rm conv}$, as shown in Figure~\ref{fig:doa_interp_conv}, display undesirable behavior; clear fade-in fade-out effects are visible, and non-negativity is violated as $\tau$ approaches $2$ due to the fact that $\bR^{\rm conv}$ becomes indefinite. Similar objections may be raised against the geodesics $\tilde{\bR}_\tau$ and $\breve{\bR}_\tau$, shown in Figures~\ref{fig:doa_interp_tilde} and \ref{fig:doa_interp_breve}, both displaying fade-in fade-out effects and thereby fail to model any displacement of the targets. Also, note that the total power of the signal varies greatly as $\tau$ goes from 0 to 2, especially for~$\breve{\bR}_\tau$.
%
%
%
%%%%% Figure: DOA interpolation, \breve{R} %%%%%%%%%%%%
\begin{figure}[t!]
        \centering
        %\vspace{-2.5 mm}
            \includegraphics[width=.48\textwidth]{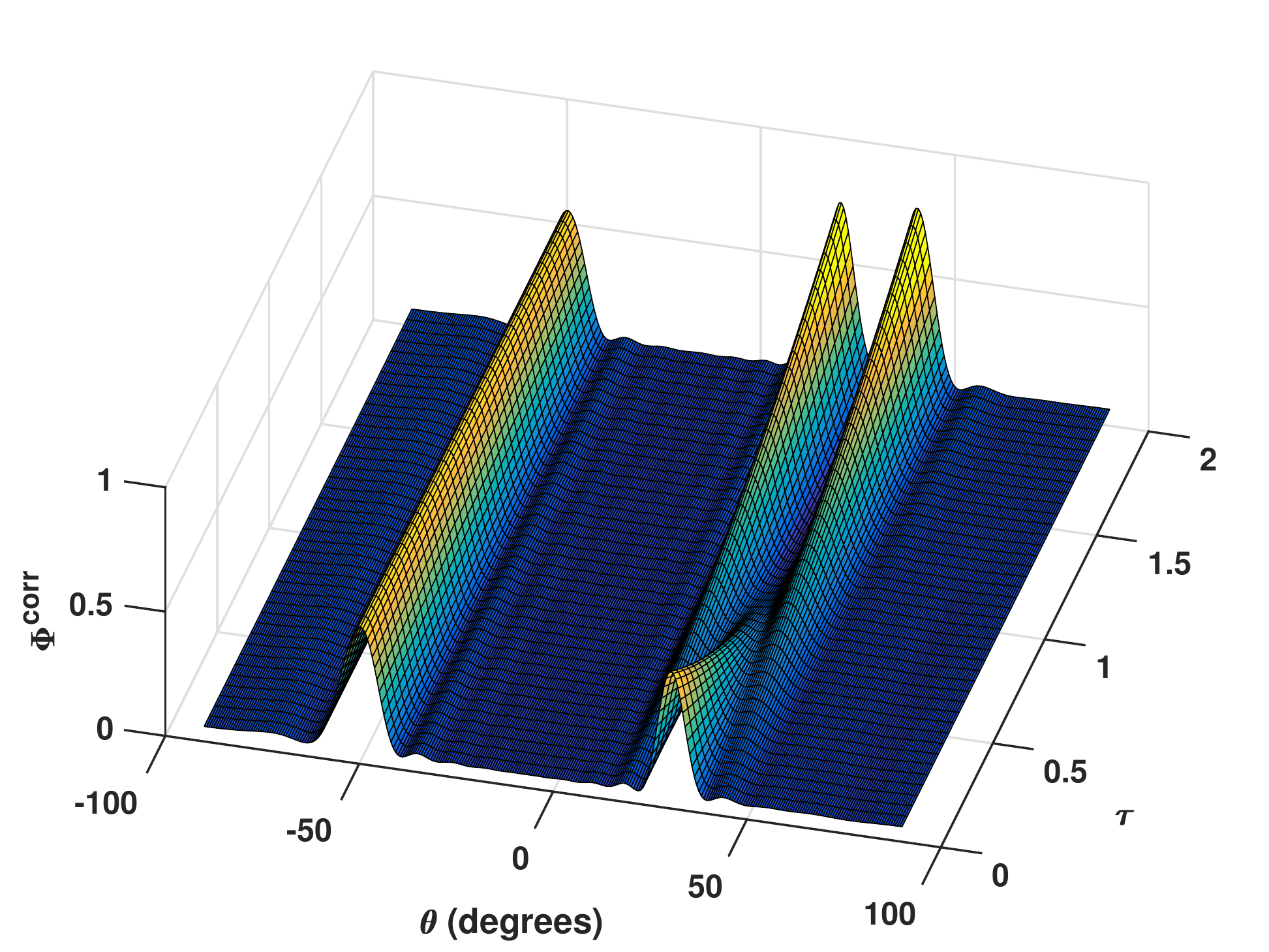}
           \caption{Interpolated spatial spectrum estimated as $\ba(\theta)^H\breve{\bR}_\tau\ba(\theta)$, where $\breve{\bR}_\tau$ is obtained from \eqref{eq:R_breve}.
           }
            \label{fig:doa_interp_breve}
\vspace{-2mm}\end{figure}
%%%%%%%%%%%%%%%%%%%%%%%%%%%%%%%%%
%
%
%
%%%%% TRACKING OF AR %%%%%%%%
\vspace{-2mm}
\subsection{Tracking of an AR-process}\label{ex:tracking}
Next, we illustrate the approach in \eqref{eq:omt_tracking}  
for the tracking of signals with slowly varying spectral content. To this end, consider a complex autoregressive (AR) process with one complex, time-varying pole.
The pole is placed at a constant radius of 0.9, and moves from the frequency $0.3\pi$ to $0.6\pi$.
Spectral estimates based directly on covariance matrix estimates $\hat{\bR}$ are shown in the top plot of Figure~\ref{fig:ar_geodesic_tsp_n15_2_both_reg_and_direct}. These covariance matrix estimates are obtained as the outer product estimate, based on 150 samples each, where the overlap between each estimate is 80 samples. Each estimated covariance matrix is of size $n = 15$. As can be seen, the spectral estimates are very noisy and vary greatly in power. Using five of these covariance matrix estimates $\hat{\bR}$, evenly spaced throughout the signal, we solve \eqref{eq:omt_tracking} with $\lambda = \frac{1}{2n^2}$ in order to obtain an estimated covariance path, $\bR_\tau$. The resulting spectra, estimated using \eqref{eq:corr_gram}, are shown in the bottom plot of Figure~\ref{fig:ar_geodesic_tsp_n15_2_both_reg_and_direct}. As can be seen, the path resulting from the proposed method allows for a smooth tracking of the shift in spectral content.
For comparison, Figure~\ref{fig:ar_geodesic_Euclid_and_logEuclid} displays the corresponding spectral estimates obtained from fitting geodesics induced by the Euclidean metric (in the top plot) as well as the log-Euclidean metric (in the bottom plot) to the same covariance estimates. Specifically, for the Euclidean case, the geodesic is constructed as $\bR_\tau = (1-\tau)\bR_0 + \tau\bR_1$, where $(\bR_0,\bR_1)$ solves 
\begin{align}
	\argminwrt[\bR_0\succeq 0, \bR_1\succeq 0] & \sum_{j=1}^J\norm{(1-\tau)\bR_0+\tau\bR_1 -\hat{\bR}_{\tau_j}}_F^2,
\end{align}
whereas for the log-Euclidean metric (see, e.g., \cite{JayasumanaHSLH13_cvpr}), the geodesic is given by \eqref{eq:R_log_euclid}, where $(\bR_0,\bR_1)$ solves
\begin{align}
	\argminwrt[\bR_0\succ 0, \bR_1\succ 0] & \sum_{j=1}^J\norm{(1\!-\!\tau)\log\left(\bR_0\right)\!+\!\tau\log\left( \bR_1\right)\!-\!\log\left(\hat{\bR}_{\tau_j}\right)}_F^2.
\end{align}
As remarked earlier, we here require $\bR_0$ and $\bR_1$ to be positive definite in order for the geodesic to be defined. It may also be noted that this approach requires all estimated covariance matrices $\hat{\bR}_{\tau_j}$ to have full rank. As may be seen in Figure~\ref{fig:ar_geodesic_Euclid_and_logEuclid}, both these fitted geodesics imply spectral estimates displaying significant fade-in fade-out effects, which should be contrasted with the proposed method's ability to here produce a reasonable and intuitive interpolation.
%
%
%
%%%%% Figure: AR, proposed+direct %%%%%%%%%%%%
\begin{figure}[t]
        \centering
        %\vspace{-.5 mm}
            \includegraphics[width=.46\textwidth]{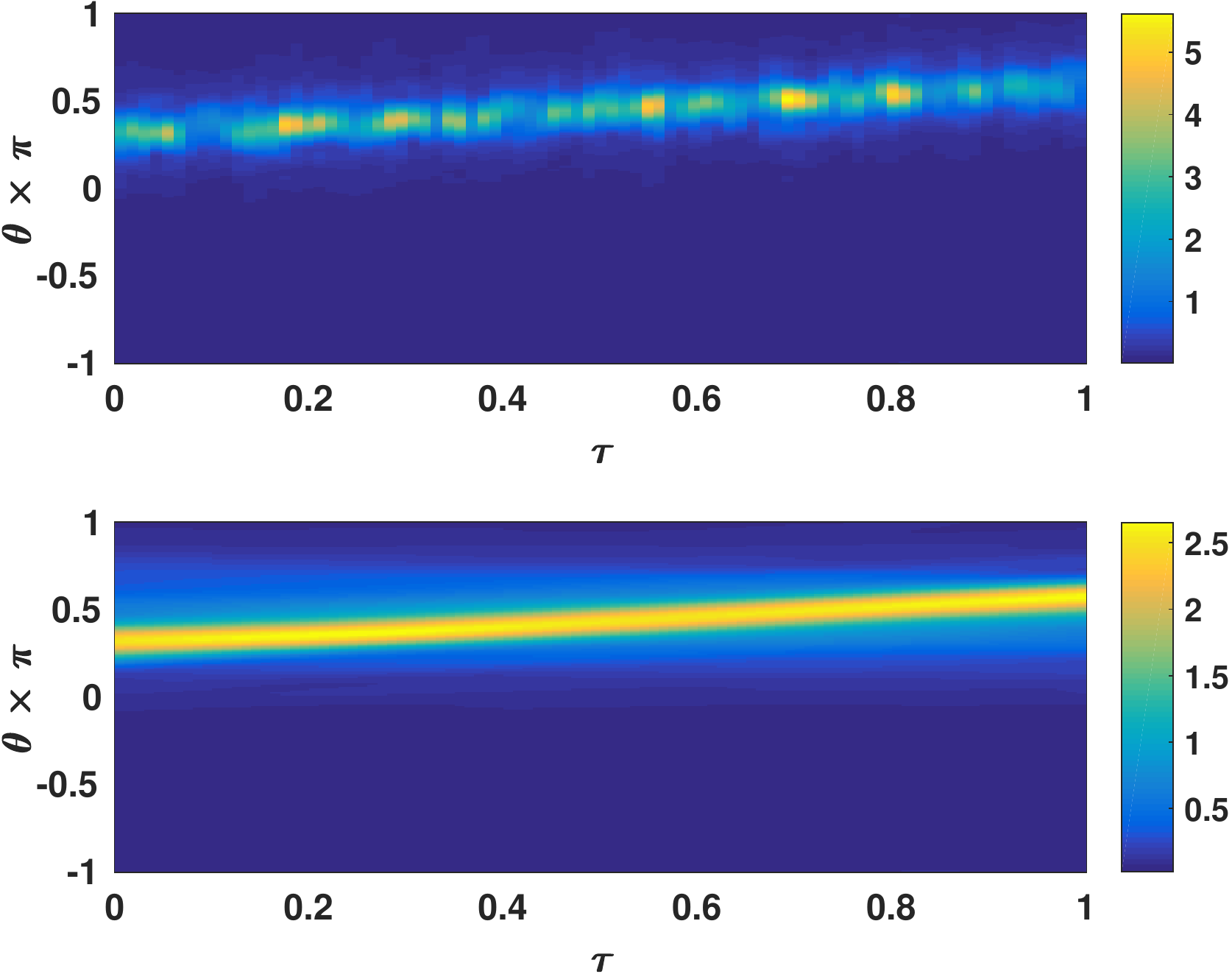}
           \caption{Spectrum estimated as $\ba(\theta)^H\hat{\bR}\ba(\theta)$, where $\hat{\bR}$ is estimated as the sample covariance matrix based on 100 samples in each window (top plot), as well as $\ba(\theta)^H\bR_\tau\ba(\theta)$, where $\bR_\tau$ is obtained by solving \eqref{eq:omt_tracking}, fitted to a sequence of five covariance estimates (bottom plot).}
            \label{fig:ar_geodesic_tsp_n15_2_both_reg_and_direct}
\vspace{-2mm}\end{figure}
%%%%%%%%%%%%%%%%%%%%%%%%%%%%%%%%%
%
%
\subsection{K-means clustering}\label{ssec:Clustering}
As a simple illustration of how to utilize the barycenter formulation in Section~\ref{sec:barycenter}, we consider the application of unsupervised clustering of phonemes. Specifically, we consider 7 utterances; 3 utterances of the phoneme \textit{/ae/}, 2 utterances of \textit{/oy/}, and 2 utterances of \textit{/n/} taken from an annotated recording sampled at 16 kHz, with the durations of the different phonemes varying between 30 ms and 174 ms. For each utterance, we estimate an $n\times n$ covariance matrix of size $n = 10$, which is a quite common covariance matrix size in speech coding applications (see, e.g., \cite{ChristensenJ09}), and then run a K-means algorithm that alternates between classifying each estimated covariance matrix according to \eqref{eq:omt_barycenter_classify} and computing new barycenters according to \eqref{eq:omt_barycenter}. To ensure that the classification is unaffected by differences in signal power, which potentially could be a dominating factor, each covariance matrix estimate is normalized as to have its diagonal elements, i.e., $r_0$, equal to one\footnote{To preserve the diagonal elements of each matrix in this setting, any \mbox{$\kappa > \max_{\theta, \varphi} c(\theta,\varphi)$} constitues a valid choice.}. We initiate the algorithm by choosing initial barycenters as convex combinations of the available covariance matrix estimates, and demand separation into three clusters. The algorithm is then run until convergence, i.e., until the classification has stabilized. 
%
%
%%%%% Figure: AR,  Euclidean and log-Euclidean %%%%%%%%%%%%
\begin{figure}[t!]
        \centering
        %\vspace{-2.5 mm}
            \includegraphics[width=.45\textwidth]{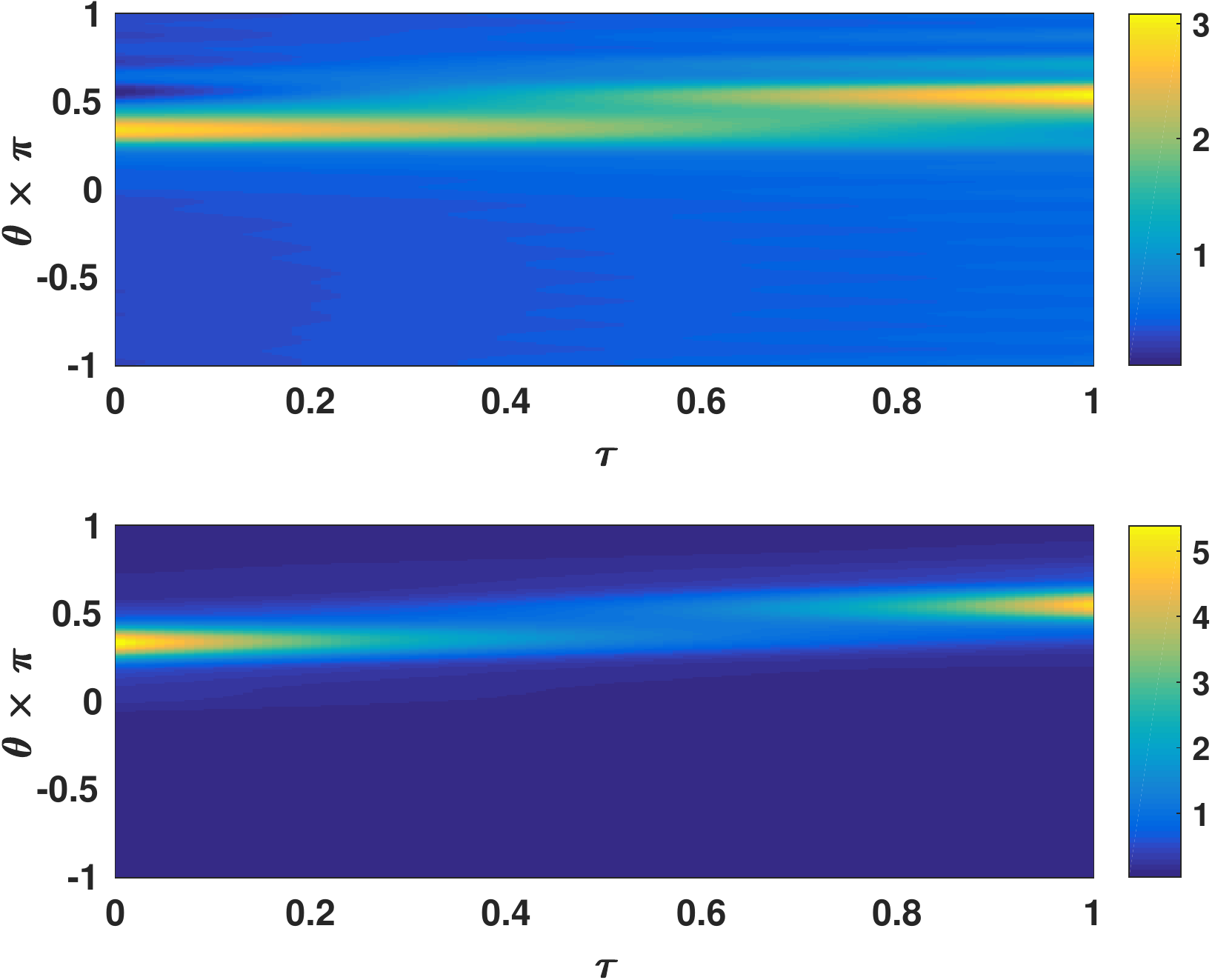}
           \caption{Spectrum estimated as $\ba(\theta)^H\bR_\tau\ba(\theta)$, where $\bR_\tau$ is the geodesic induced by the Euclidean metric (top plot) and the log-Euclidean metric (bottom plot), when fitted to a sequence of five covariance estimates.}
            \label{fig:ar_geodesic_Euclid_and_logEuclid}
\vspace{-2mm}\end{figure}
%%%%%%%%%%%%%%%%%%%%%%%%%%%%%%%%%
%
%
%
As a comparison, we perform the same K-means clustering using the Euclidean and the log-Euclidean metrics as described earlier, as well as the Kullback-Leibler divergence and the ellipticity distance measure introduced in \cite{OllilaSTW16_arXiv} in order to compute distances as well as barycenters.
Specifically, for two positive definite matrices $\bR_0$ and $\bR_1$ of size $n$, the Kullback-Leibler divergence is given by
\begin{align*}
	d_{\rm KL}(\bR_0,\bR_1)=\trace{\bR_0^{-1}\bR_1}-\log\abs{\bR_0^{-1}\bR_1}-n,
\end{align*}
and the ellipticity distance is given by
\begin{align*}
	d_{\rm E}(\bR_0,\bR_1)=n\log\left( \frac{1}{n}\trace{\bR_0^{-1}\bR_1}\right)-\log\abs{\bR_0^{-1}\bR_1}.
\end{align*}
As may be noted, these distance notions are only defined for non-singular matrices $\bR_0$ and $\bR_1$. For a set of $J$ observed covariances $\bR_j$, $j = 1,\ldots,J$, the barycenter induced by the Kullback-Leibler distance is given by
\begin{align*}
	\bR_{\rm KL}^{\text{bary}} = \left( \frac{1}{J} \sum_{j=1}^J \bR_j^{-1} \right)^{-1},
\end{align*}
whereas for the ellipticity distance, the barycenter is the solution to the fixed point equation (see \cite{OllilaSTW16_arXiv})
\begin{align*}
	\bR_{\rm E}^{\text{bary}} = \left( \frac{n}{J} \sum_{j=1}^J \frac{\bR_j^{-1}}{\trace{\bR_j^{-1} \bR_{\rm E}^{\text{bary}}}}  \right)^{-1}.
\end{align*}
The solution $\bR_{\rm E}^{\text{bary}}$ is unique only up to a positive scaling factor, and thus, we here normalize $\bR_{\rm E}^{\text{bary}}$ as to have unit diagonal.
It may be noted, that even for Toeplitz $\bR_j$, $j = 1,\ldots, J$, the barycenters $\bR_{\rm KL}^{\text{bary}}$ and $\bR_{\rm E}^{\text{bary}}$ are not Toeplitz in general.
As the K-means algorithm may converge to different solutions, i.e., different sets of clusters, depending on the choice of initial points, we have for each choice of distance measure run the algorithm several times using different starting points and selected the solution corresponding to the least total distance between each barycenters and its assigned covariance matrices. The results are shown in Table~\ref{table:classified_phonemes}. 
As can be seen, the proposed distance notion $\Omtcov_\kappa$ produces a clustering in which the third instance of the \textit{/ae/} phoneme is erroneously grouped together with the two utterances of \textit{/oy/}; apart from this, the clustering corresponds well to the true phonemes. The comparison distance measures do in this example produce clusterings that differ more from the ground truth than the clustering using the proposed distance notion. It may also be noted that the Euclidean and log-Euclidean produce the same clustering.
Table~\ref{table:phonemes_distances} presents the corresponding distance matrix, i.e., the matrix detailing the distance $\Omtcov_\kappa$ between each covariance matrix estimate and each barycenter. Note that for each utterance, the distances have been normalized by the least distance for that utterance. It is worth noting that the clusters are quite well-separated.
Although being limited  in scope, the example illustrates that the proposed distance notion may indeed be used in order to perform clustering and classification of stochastic processes based on their estimated covariances.
%
%
%%%%% TABLE: PHONEME CLASSES %%%%%%%%
\begin{table}[t!]
\centering
\begin{tabular}{c|c|c|c|c|c|c|c}
\toprule
\toprule
Utterance & 1 & 2 & 3 & 4 & 5 & 6 & 7 \\
Phoneme & \textit{/ae/} & \textit{/ae/}& \textit{/ae/} &\textit{/oy/} & \textit{/oy/} & \textit{/n/} & \textit{/n/}  \\ 
\toprule
Proposed & 1 & 1 & 2 & 2 & 2 & 3 & 3 \\ 
Euclidean & 1 & 1 & 1 & 1 & 2 & 2 & 3 \\
Log-Euclidean & 1 & 1 & 1 & 1 & 2 & 2 & 3 \\
KL divergence & 1 & 1 & 2 & 1 & 3 & 3 & 3 \\
Ellipticity & 1 & 1 & 2 & 1 & 3 & 3 & 3 \\
\bottomrule
\bottomrule
\end{tabular}
\vspace{2mm}
\caption{Clustering of 7 utterances into three clusters using a K-means algorithm utilizing the barycenter formulation in Section~\ref{sec:barycenter}, as well as four comparison distance measures. The third to seventh rows indicate the identified classes as given by the algorithm using the different distance measures.} \label{table:classified_phonemes}
\vspace{-2mm}
\end{table}
%%%%%%%%%%%%%%%%%%%%
%
%
%%%%% TABLE: PHONEMES, DISTANCE MATRIX %%%%%%%%
\begin{table}[t]
\vspace{5mm}
\centering
{\renewcommand{\arraystretch}{1}
\begin{tabular}{c|c|c|c|c|c|c|c}
\toprule
\toprule
Utterance & 1 & 2 & 3 & 4 & 5 & 6 & 7 \\
Class 1 &1 &  1 &  4.20 &  7.10 &  5.02  & 10.54  & 18.67  \\ 
Class 2 & 8.79 &   7.83  &  1 &  1 &   1  & 4.70  &  9.99 \\
Class 3 & 20.51  &  24.93 &   8.30  &  12.05  &  1.40  &  1  &  1 \\
\bottomrule
\bottomrule
\end{tabular}}
\vspace{4mm}
\caption{Distance, as measured by $\Omtcov_\kappa$, between each utterance and each barycenter for the three identified clusters. Each distance has been normalized by the least distance for each utterance.} \label{table:phonemes_distances}
\vspace{-2mm}
\end{table}
%%%%%%%%%%%%%%%%%%%%
%
%
%
\subsection{Fixed cost for moving mass}
In some scenarios, there might be a relatively large noise component present in both $\bR_0$ and $\bR_1$, where the noise power is localized in frequency and this localization is the same in both $\bR_0$ and $\bR_1$. This might be the case in, e.g., DOA estimation scenarios where a source of interest is moving in the presence of a stationary interferer. Then, if the source for example moves past the location of the interferer, the optimal transport problem in \eqref{eq:omt_toeplitz2} may result in a power association such that the source and the interferer are mixed together. Such a scenario is illustrated in Figure~\ref{fig:doa_interferer}, showing estimated spectra obtained from the interpolant $\bR_\tau$, as given by \eqref{eq:omt_toeplitz2} with the cost function $\costfunc(\theta,\varphi) = \abs{e^{i\theta}-e^{i\varphi}}^2$. Here, the source is moving from look-angle $\theta = 30\si{\degree}$ to $\theta = -20\si{\degree}$, whereas there is a fixed interferer located at $\theta = 0\si{\degree}$ having a third of the power of the source.
%
%
%
%
%%%%%%%%%%%%%%%%%%%%%%%%%%%%%%%%%
%
 %%%%%% Figure: DOA doa_interferer %%%%%%%%%%%%
\begin{figure}[t!]
        \centering
        %\vspace{-2.5 mm}
            \includegraphics[width=.48\textwidth]{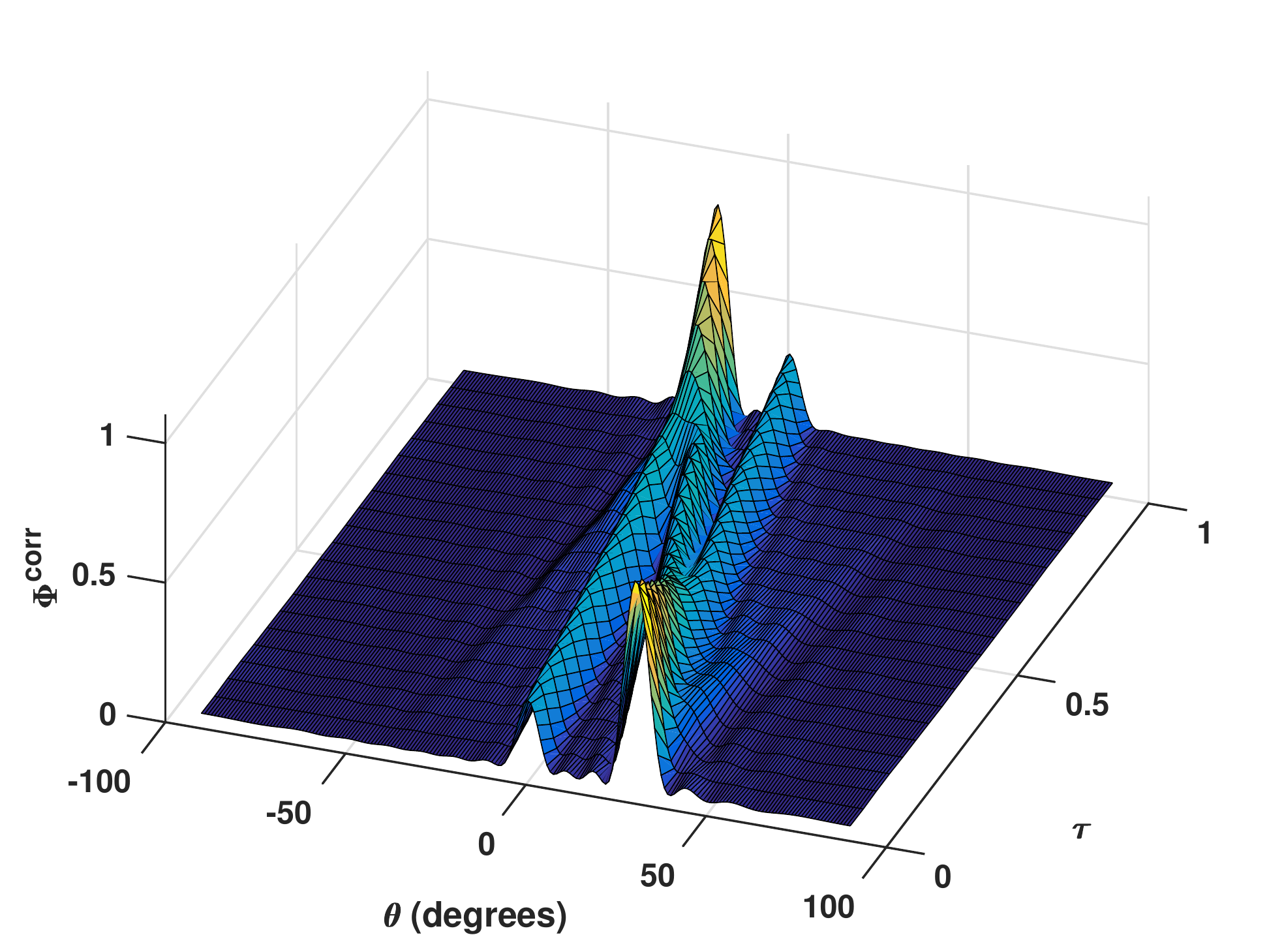}
           \caption{Interpolated spatial spectrum estimated as $\ba(\theta)^H\bR_{\tau}\ba(\theta)$, where $\bR_{\tau}$ is obtained by solving \eqref{eq:omt_toeplitz2}, for the case of one moving source and one static interferer, using with the cost function $\costfunc(\theta,\varphi) = \abs{e^{i\theta}- e^{i\varphi}}^2$.%with the covariances $\bR_0$ and $\bR_1$.
           }
            \label{fig:doa_interferer}
\vspace{-2mm}\end{figure}
%%%%%%%%%%%%%%%%%%%%%%%%%%%%%%%
%
%
 %%%%%% Figure: DOA doa_interferer_fixed_cost %%%%%%%%%%%%
\begin{figure}[t!]
        \centering
        \vspace{5 mm}
            \includegraphics[width=.45\textwidth]{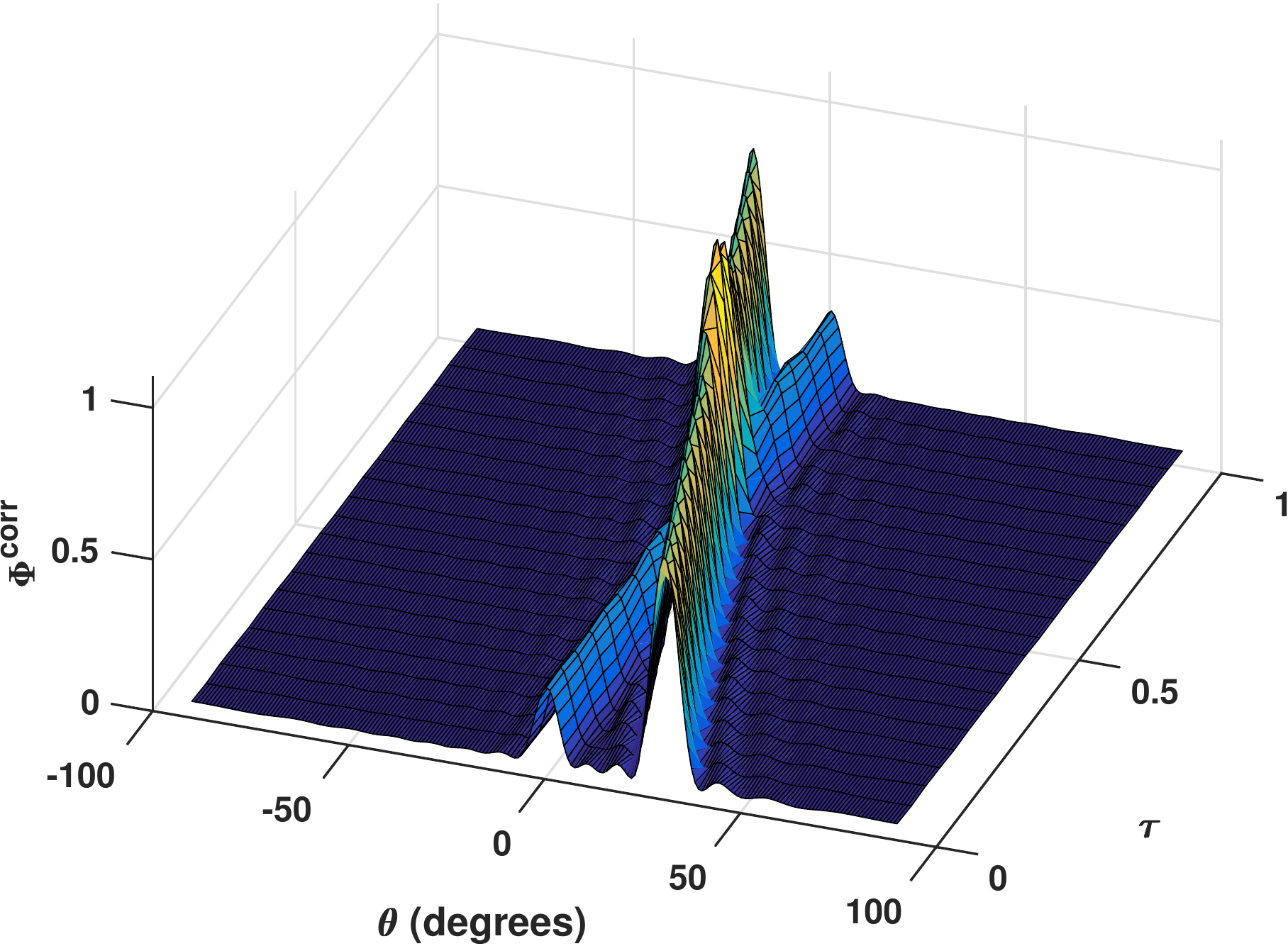}
           % \vspace{2mm}
           \caption{Interpolated spatial spectrum estimated as $\ba(\theta)^H\bR_{\tau}\ba(\theta)$, where $\bR_{\tau}$ is obtained by solving \eqref{eq:omt_toeplitz2}, for the case of one moving source and one static interferer, using with the modified cost function in \eqref{eq:mod_cost}.%with the covariances $\bR_0$ and $\bR_1$.
           }
            \label{fig:doa_interferer_fixed_cost}
\vspace{-2mm}
\end{figure}
%%%%%%%%%%%%%%%%%%%%%%%%%%%%%%%%%
%
%
In order to avoid this type of problem and promoting transport plans that avoid transporting stationary masses, the cost function may be modified according to
\begin{align} \label{eq:mod_cost}
\costfunc(\theta, \varphi)=\begin{cases} 1+|e^{i\theta}- e^{i\varphi}|^2 & \mbox{if }\theta\neq \varphi\\
0 & \mbox{if } \theta= \varphi.
\end{cases}
\end{align}
Thus, the cost function is here formulated such that there is a fixed, baseline cost of moving any mass. The resulting estimated spectra obtained from the interpolant $\bR_\tau$ resulting from solving \eqref{eq:omt_toeplitz2} using this modified cost function is shown in Figure~\ref{fig:doa_interferer_fixed_cost}. As can be seen, the source and the interferer are now well separated throughout the interpolated path.
\section{Conclusions and further directions}\label{sec:conclusions}
In this work, we have proposed a notion of distance for positive semi-definite Toeplitz matrices. By considering spectral representations of such matrices, the proposed measure is based on distances, in an optimal mass transport sense, between families of spectra consistent with the Toeplitz matrices. We have shown that the proposed distance measure, under some mild assumptions, is contractive with respect to additive and multiplicative noise. The proposed measure may be used to, for example, define inter- and extrapolation of Toeplitz matrices, being of interest in applications such as tracking of slowly varying signals. 

A future direction of this methodology is to generalize the distance measure for structured matrices such as Toeplitz-block-Toeplitz matrices and input-to-state covariances in the THREE framework for spectral analysis \cite{ByrnesGL00_48}.
The latter generalizes ideas from beamspace processing, enabling 
the user to improve resolution and robustness in power spectral estimation over selected frequency bands  \cite{AminiG06_54, KarlssonG13_58}.
This will be the subject of further research.
%
%
%
%%%%%%%%%%%%%% APPENDIX %%%%%%%%%%
\appendix
\subsection{Proof of Proposition Proposition \ref{prop:semi-metric}}
Positivity and symmetry of $\Omtcov_\kappa$ follows directly from positivity and symmetry of $\costf$. Further, is clear that $\Omtcov_\kappa(\bR_0,\bR_1) = 0$ if $\bR_0 = \bR_1$. Thus, it remains to show that 
$\Omtcov_\kappa(\bR_0,\bR_1) = 0$ implies that $\bR_0 = \bR_1$. Next, if the objective function 
\eqref{eq:omt_toeplitz_kappa2} is equal to zero, then the first term is zero and since $c$ is a semi-metric, the transport plan $M$ only has support on $\theta=\varphi$. Since the second and third terms of \eqref{eq:omt_toeplitz_kappa2} are zero, it follows that $\Psi_0=\Psi_1$, and hence $\bR_0 = \bR_1$.
\hfill$\square$
%
%
%%%%% PROOF OF PROPOSITION, Additive NOISE %%%%%%%%
\subsection{Proof of Proposition \ref{prop:contractive_prop} } \label{appendix:proof_add}
First, we show that $\Omtcov_\kappa$ is contractive with respect to additive noise.
Consider two processes with covariance matrices $\bR_0$ and $\bR_1$ and assume that they are both additively corrupted by an independent noise process with covariance $\bR_w$. This results in processes with covariances $\bR_0' = \bR_0 + \bR_w$ and \mbox{$\bR_1' = \bR_1 + \bR_w$}, respectively. From the dual formulation in \eqref{eq:dual_problem}, the distance between these covariance matrices is given by
\begin{align*}
	\Omtcov_\kappa(\bR_0',\bR_1')=&\max_{\left(\bLambda_0,\bLambda_1\right) \in \Omega_\costfunc^\kappa}  \langle \bLambda_0, \bR_0  \rangle + \langle \bLambda_1, \bR_1  \rangle \\
	& \qquad\qquad\qquad+  \langle \bLambda_0 \!+\! \bLambda_1, \bR_w  \rangle.
\end{align*}
Let $\Phi_w$ be any spectrum consistent with the noise covariance $\bR_w$, i.e., with $\Gamma(\Phi_w) = \bR_w$.  We then have that
\begin{align*}
	\langle \bLambda_0 + \bLambda_1, \bR_w  \rangle &= \langle \bLambda_0 + \bLambda_1, \Gamma(\Phi_w) \rangle \\&= \langle \Gamma^*(\bLambda_0) + \Gamma^*(\bLambda_1), \Phi_w \rangle\\
	&= \int_{\RT} \left( \Gamma^*(\bLambda_0)(\theta) + \Gamma^*(\bLambda_1)(\theta) \right) \Phi_w(\theta) d\theta.
\end{align*}
As $\left(\bLambda_0,\bLambda_1\right) \in \Omega_\costfunc^\kappa$, it holds that
\begin{align*}
	\Gamma^*(\bLambda_0)(\theta) + \Gamma^*(\bLambda_1)(\theta) \leq \costfunc(\theta,\theta) = 0 \mbox{ for all } \theta \in \RT
\end{align*}
and since $\Phi_w\ge 0$, we get
\begin{align*}
	\langle \bLambda_0 + \bLambda_1, \bR_w  \rangle \leq 0.
\end{align*}
Hence, it follows that
\begin{align*}
	\Omtcov_\kappa(\bR_0',\bR_1')\!\le\!\max_{\left(\bLambda_0,\bLambda_1\right) \in \Omega_\costfunc}\!\langle \bLambda_0, \bR_0  \rangle\!+\!\langle \bLambda_1, \bR_1  \rangle\!=\!\Omtcov_\kappa(\bR_0,\bR_1).
\end{align*}
Next, we show that $\Omtcov_\kappa$ is contractive with respect to multiplicative noise.
Let the noise covariance matrix be $\bR_w$, implying that the covariances of the contaminated processes are \mbox{$\bR_0' = \bR_0 \odot \bR_w$} and $\bR_1' = \bR_1 \odot \bR_w$, respectively. Let $\Phi_w$ be any spectrum consistent with $\bR_w$, i.e., $\Gamma(\Phi_w) = \bR_w$. Also, let the diagonal elements of $\bR_w$ be smaller than or equal to 1, so that
\begin{align*}
	\int_\RT \Phi_w(\theta) d\theta \leq 2\pi.
\end{align*}
We have
\begin{align*}
	\Omtcov_\kappa(\bR_0',\bR_1')=&\max_{\left(\bLambda_0,\bLambda_1\right) \in \Omega_\costfunc^\kappa} \quad \langle \bLambda_0, \bR_0\!\odot\!\bR_w  \rangle\!+\!\langle \bLambda_1, \bR_1\!\odot\!\bR_w  \rangle \\
	=& \max_{\left(\bLambda_0,\bLambda_1\right) \in \Omega_\costfunc^\kappa} \quad \langle \bLambda_0\!\odot\!\overline{\bR_w}, \bR_0  \rangle\!+\!\langle \bLambda_1\!\odot\!\overline{\bR_w}, \bR_1 \rangle \\
	=& \max_{\substack{\left(\bLambda_0,\bLambda_1\right) \in \Omega_\costfunc^\kappa\\ \tilde{\bLambda}_0,\tilde{\bLambda}_1 \in \RM^{n}}} \;\; \langle \tilde{\bLambda}_0, \bR_0  \rangle + \langle \tilde{\bLambda}_1, \bR_1 \rangle \\
	& \quad\mbox{subject to} \quad\quad \tilde{\bLambda}_0 = \bLambda_0 \odot \overline{\bR_w}\\
	& \quad\quad\quad \qquad\qquad \tilde{\bLambda}_1 = \bLambda_1 \odot \overline{\bR_w}.
\end{align*}
In order to show that $\left(\bLambda_0,\bLambda_1\right) \in \Omega_\costfunc^\kappa$ implies $\left(\tilde{\bLambda}_0,\tilde{\bLambda}_1\right) \in \Omega_\costfunc^\kappa$, we note the following.
Assume that $\left(\bLambda_0,\bLambda_1\right) \in \Omega_\costfunc^\kappa$. Then, 
\begin{align} \label{eq:intermediate_result_multnoise}
	\Gamma^*(\bLambda_0)(\theta-\phi)\!+\!\Gamma^*(\bLambda_1)(\varphi-\phi)\!\leq\!\costfunc(\theta-\phi,\varphi-\phi)\!=\!\costfunc(\theta,\varphi) 
\end{align}
for all $\phi, \theta,\varphi \in \RT$, where the last equality follows from the assumption that $c$ is shift-invariant. Define
\begin{align*}
	\breve{\Phi}_w(\phi) = \Phi_w(-\phi)
\end{align*}
which satisfies $\int \breve{\Phi}_w(\phi)d\phi \leq 2\pi$ and $\breve{\Phi}_w(\phi)\geq 0$, $\forall \phi \in \RT$.
Multiplying both sides of \eqref{eq:intermediate_result_multnoise} with $\breve{\Phi}_w(\phi)$ and integrating with respect to $\phi$ then yields
\begin{align*}
	\frac{1}{2\pi}\left( \Gamma^*(\bLambda_0) * \breve{\Phi}_w \right)(\theta) + \frac{1}{2\pi}\left( \Gamma^*(\bLambda_1) * \breve{\Phi}_w \right)(\varphi) \leq \costfunc(\theta,\varphi)
\end{align*}
for all $\theta,\varphi \in \RT$, where both sides have been divided by $2\pi$. Similarly,
\begin{align}
	\frac{1}{2\pi}\left( \Gamma^*(\bLambda_0) * \breve{\Phi}_w \right)(\theta) \leq \kappa ,\; \frac{1}{2\pi}\left( \Gamma^*(\bLambda_1) * \breve{\Phi}_w \right)(\varphi) \leq \kappa
\end{align}
for all $\theta,\varphi \in \RT$. Further, for any $\theta \in \RT$, we have
\begin{align*}
	\frac{1}{2\pi}\!\left( \Gamma^*(\bLambda_0)\!*\!\breve{\Phi}_w \right)\!(\theta)\!&=\frac{1}{4\pi^2}\!\int_\RT \ba(\varphi)^H\bLambda_0\ba(\varphi)\breve{\Phi}_w(\theta-\varphi) d\varphi \\
	&=\frac{1}{4\pi^2}\tr\!\left(\!\bLambda_0\!\int_\RT\!\ba(\varphi)\ba(\varphi)^H\breve{\Phi}_w(\theta-\varphi) d\varphi\!\right) \\
	&=\frac{1}{4\pi^2}\langle \bLambda_0,  \int_\RT  \ba(\varphi)\ba(\varphi)^H\Phi_w(\varphi - \theta) d\varphi \rangle \\
	&=\frac{1}{4\pi^2}\langle \bLambda_0,  \int_\RT  \ba(\varphi)\ba(\varphi)^H\left(\Phi_w *\delta_\theta  \right)(\varphi)d\varphi \rangle \\
	&=\frac{1}{2\pi}\langle \bLambda_0, \Gamma\left(\Phi_w *\delta_\theta  \right) \rangle,
\end{align*}
where the last equality uses the definition of the operator $\Gamma(\cdot)$. As $\delta_\theta(\varphi)$ is a spectrum consistent with the rank-one covariance matrix $\bR_\theta = \frac{1}{2\pi}\ba(\theta)\ba(\theta)^H$, we have, by the properties of the Fourier transform, that $\left(\Phi_w *\delta_\theta  \right)(\varphi)$ is a spectrum consistent with the covariance matrix $2\pi\bR_w \odot \bR_\theta$, i.e.,
\begin{align*}
	\frac{1}{2\pi}\langle \bLambda_0, \Gamma\left(\Phi_w *\delta_\theta  \right) \rangle = \langle  \bLambda_0, \bR_w \odot \bR_\theta \rangle,
\end{align*}	
 and therefore
\begin{align*}
	\frac{1}{2\pi}\left( \Gamma^*(\bLambda_0) * \breve{\Phi}_w \right)(\theta) &= \langle  \bLambda_0,\bR_w \odot \bR_\theta \rangle \\
	&=\langle  \bLambda_0 \odot \overline{\bR_w}, \bR_\theta \rangle \\
	&= \langle  \tilde{\bLambda}_0, \bR_\theta \rangle \\
	&=\langle  \tilde{\bLambda}_0, \frac{1}{2\pi}\ba(\theta)\ba(\theta)^H \rangle \\
	&= \frac{1}{2\pi}\ba(\theta)^H\tilde{\bLambda}_0 \ba(\theta) \\
	&= \Gamma^*(\tilde{\bLambda}_0)(\theta),
\end{align*}
where the third equality follows from the definition of $\tilde{\bLambda}_0$ and the last from the definition of $\Gamma^*(\cdot)$. Using the same reasoning for $\bLambda_1$, we thus have that $\left(\bLambda_0,\bLambda_1\right) \in \Omega_\costfunc^\kappa$ implies \mbox{$\left(\tilde{\bLambda}_0,\tilde{\bLambda}_1\right) \in \Omega_\costfunc^\kappa$}. Combined, this yields 
{
\small{
\begin{align*}
	\Omtcov_\kappa(\bR_0',\bR_1')\leq \max_{\left(\tilde{\bLambda}_0,\tilde{\bLambda}_1\right) \in \Omega_c^\kappa}  \langle \tilde{\bLambda}_0, \bR_0  \rangle + \langle \tilde{\bLambda}_1, \bR_1 \rangle \; = \Omtcov_\kappa(\bR_0,\bR_1).
\end{align*}
}
}
\hfill $\square$
\bibliographystyle{IEEEbib}
\bibliography{ElvanderJK18_manuscript_arxiv2.bbl,IEEEabrv}

\end{document}